\documentclass{article}

    \PassOptionsToPackage{numbers, compress}{natbib}

\usepackage{enumitem}
\usepackage{wrapfig}
\usepackage{amsthm}
\newtheorem*{proposition*}{Proposition}
\newtheorem*{theorem*}{Theorem}
\usepackage{algorithm}
\usepackage{algpseudocode}
\usepackage{float}        

\usepackage{mathtools}

 \usepackage[preprint]{neurips_2026}

\usepackage[utf8]{inputenc} 
\usepackage[T1]{fontenc}    
\usepackage{hyperref}       
\usepackage{url}            
\usepackage{booktabs}       
\usepackage{amsfonts}       
\usepackage{nicefrac}       
\usepackage{microtype}      
\usepackage{xcolor}         
\usepackage{amsmath}

\usepackage{amssymb}

\usepackage{tikz}
\usetikzlibrary{arrows.meta} 
\usetikzlibrary{positioning}
\usetikzlibrary{shapes.geometric}

\newtheorem{theorem}{Theorem}[section]
\newtheorem{proposition}[theorem]{Proposition}

\newtheorem{remark}[theorem]{Remark}

\author{%
  Binyamin Perets \\
  Technion -- Israel Institute of Technology \\
 \And
Shie Mannor \\
Technion -- Israel Institute of Technology \\
NVIDIA \\
}

\title{Controlling False Discovery in Arbitrarily Structured Hypothesis Spaces via Reproducing Kernels.}

\begin{document}

\maketitle

\begin{abstract}
Large-scale hypothesis testing is central to modern science, where controlling the False Discovery Rate (FDR) has become the standard approach to managing false positives across many simultaneous tests. Hypotheses rarely exist in isolation; they often exhibit structure through proximity, connectivity, or hierarchy. This structure represents both a challenge and an opportunity: while classical methods treat these dependencies as obstacles requiring conservative correction, leveraging them can substantially increase discovery power.
Here, we reframe structured FDR control as a regularized learning problem. By optimizing within a suitable Reproducing Kernel Hilbert Space (RKHS), we introduce a framework that unifies continuous domains, graphs, and hierarchies under a single algorithm through kernel choice alone. This formulation enables smooth solutions in place of the piecewise-constant fits of prior methods, principled likelihood-based hyperparameter selection rather than heuristic tuning, and inference at unobserved locations which in turn supports sample-efficient experimental design. Building on this estimator, we provide two decision rules which we prove to control the FDR. We validate our method on two sources: spatial locations derived from high-dimensional real-world datasets, and a differential gene expression task utilizing protein-protein interaction graphs.

\end{abstract}

\section{Introduction}
The challenge of multiple hypothesis testing is one of the most fundamental in science, addressing our ability to control the probability of false discoveries that arise purely from random chance when evaluating a large number of hypotheses. In practice, this often translates to controlling the False Discovery Rate (FDR). In most domains, it is more the rule than the exception to encounter multiple hypothesis testing problems where hypotheses are not independent, forming a structured hypothesis space: nearby brain regions activate together, genes in the same pathway co-regulate, and spatially proximate cells share molecular signatures. In practice, this structure significantly affects FDR results \citep{Schwartzman2011}. While the hypothesis testing literature uses several related terms inconsistently, we frame our work under the \textbf{structured} hypothesis space setup, i.e., with geometric or topological organization induced by \textbf{relations} between hypotheses (any measurable association, from spatial proximity to hierarchical nesting, that indicates hypotheses ``belong together'' in some sense). This differs from approaches that require explicit \textbf{dependencies}, i.e., formal probabilistic relationships specified through joint or conditional distributions (e.g., positive regression dependency, correlation structures), which need not be present for relations or structure to be meaningful. While classical FDR control methods typically aim to provide guarantees under unknown or arbitrary dependencies \citep{Benjamini2001} (see Section \ref{sec:related_work}), recent work has shown that when hypotheses 
exhibit known structure, this information can be leveraged for 
substantial power gains \citep{Tibshirani_2011, Cai2011}. This represents a conceptual paradigm shift: rather than treating dependencies as obstacles requiring conservative correction, these methods exploit known structure as prior information to enhance discovery while maintaining rigorous FDR control. However, such dependency-based approaches face fundamental 
limitations. They require that the structure be expressed through formal probabilistic dependencies and to be sparse. As demonstrated by \citet{Cai2011} and \citet{Heller2020}, violating either condition renders the problem computationally intractable or statistically infeasible, limiting 
applicability to settings with both known probabilistic form and low-degree dependency graphs. An alternative regularization-based approach operates on relational information without requiring probabilistic interpretation. The leading example, SmoothFDR \citep{tansey2016falsediscoveryratesmoothing}, encodes relations through graph connectivity and uses Total Variation regularization to enforce piecewise-constant significance patterns over connected components. This relaxes the probabilistic requirement but introduces new limitations: the graph structure must be specified a priori, solutions are constrained to piecewise-constant functions poorly suited to smoothly varying phenomena, and the framework lacks principled mechanisms for hyperparameter selection or interpolation to unobserved locations. 

In this paper, we address the general problem by assuming a continuous hypothesis space defined over a $d$-dimensional domain, subsequently demonstrating that almost any discrete problem can be approximated under 
this framework. By modeling the spatially varying structure within a 
Reproducing Kernel Hilbert Space (RKHS) \cite{Aronszajn1950}, diverse 
structures, from spatial grids and phylogenetic trees to protein interaction 
networks, can be incorporated through kernel choice alone. Moreover, we move beyond the restrictive piecewise-constant assumptions of current $L_1$-based methods, enabling flexible smoothness regularization and the generation of continuous significance maps valid at both observed and unobserved locations. Our framework is computationally efficient, scaling as $O(N^2)$ via natural gradient optimization, and distinguishes itself by supporting valid likelihood-based cross-validation through proving the model is properly normalized for all hyperparameter choices.

The remainder of this paper is organized as follows. Section~\ref{sec:related_work} provides preliminary definitions and positions the approach within the literature. Section~\ref{sec:problem_setup} establishes the problem formulation and introduces the spatially-varying mixture model. Section~\ref{sec:pointwise} details our point-wise optimization framework in Reproducing Kernel Hilbert Spaces (RKHS). Section~\ref{sec:decision_rules} translates the resulting estimates into two decision rules with formal FDR guarantees and characterizes their relative power. Section~\ref{sec:entire_domain} extends inference to the entire domain and addresses optimal experimental design for untested locations.
Section~\ref{sec:practical} discusses practical implementation, including likelihood-based hyperparameter selection and the application of our framework to discrete structures via graph kernels. Finally, Section~\ref{sec:experiments} evaluates the method on semi-synthetic benchmarks and on two real-world tasks: HIGGS particle physics and TCGA gene expression on the STRING protein-protein interaction network.

\section{Preliminaries, Related Works, and Positioning}
\label{sec:related_work}

\textbf{p-value} is the probability under the null of obtaining a test statistic at least as extreme as the one observed. \textbf{False Discovery Rate (FDR)} is the expected proportion of false positives in the multiple hypotheses scenario:$FDR = E \left[ \frac{V}{R} \mid R \geq 0 \right] P(R>0)$ where $V$ is the number of false positives and $R$ is the total number of rejections. 
FDR control literature distinguishes between global and \textbf{local FDR} ($lfdr$): $lfdr$ is a Bayesian approach providing per-hypothesis significance by estimating the posterior null probability \citep{Efron2004} using a two-group mixture with marginal density: $f(z) = \alpha_0 f_0(z) + (1-\alpha_0) f_1(z).$  Here, $\alpha_0$ and $(1-\alpha_0)$ are the prior probabilities of the null and alternative hypotheses, respectively, $f_0(z)$ is the null distribution (theoretically uniform for p-values, though often empirically estimated), and $f_1(z)$ is the alternative distribution (typically unknown but assumed to be skewed toward zero). Given this mixture model, the local-fdr for test statistic $z$ is defined as the posterior probability: $\text{lfdr}(z) = P(H_0 | Z=z) = \frac{\alpha_0 f_0(z)}{\alpha_0 f_0(z) + (1-\alpha_0) f_1(z)}$.

\subsection{Related Works and Positioning}
Controlling the FDR for related hypotheses is a long-standing challenge. Classical procedures, such as the Benjamini-Hochberg (B-H) method \citep{B-H}, control the FDR under independence or positive regression dependency, while other \citep{Barber_2015} are required for arbitrary relations. These methods treat the structure as a nuisance to be corrected for. Recent work leverages relations for power gains, yielding structured FDR methods that fall into several categories:
\textbf{Explicit Dependency Modeling} models dependencies as proper statistical relations (e.g., joint or conditional distributions), and have shown that this can lead to optimal \textit{"oracle"} procedures that significantly boost power \citep{Cai2011}. However, those methods  \citep{Tibshirani_2011,Heller2020} typically require restrictive conditions: (1) the dependencies must represent actual probabilistic relationships with dependency graph given explicitly; this lacks mechanisms for handling continuous domains, (2) these dependencies must be sparse as noted by \cite{Heller2020}.
\textbf{Adaptive P-value Methods.} 
Methods such as LAWS \citep{Cai2021} and STRAW \citep{wang2023strawstructureadaptiveweightingprocedure} addresses spatial dependencies through adaptive p-value weighting, constructing local weights from spatial neighbors using discrete windows or local weighted averaging to up-weight or down-weight p-values before applying standard FDR correction. Importantly, these approaches remain heuristic: lacking of explicit model of the spatial structure with no clear way to enforce prior beliefs or a principled mechanism for hyperparameter selection. On the same line, \textbf{but much closer to our framework}, AdaFDR \cite{Zhang2019} and NeuralFDR \cite{xia2017neuralfdrlearningdiscoverythresholds} leverages features associated with each hypothesis through learned reweighting functions which are learned over the features (different from our "structure" that define the relations between hypotheses). Similarly, these methods suffer the same drawbacks, do not allow for explicit structure or priors, and provide no tractability for the bias introduced.
\textbf{Regularization-Based Approaches} focuses on regularization-based approaches that treat structure as a smoothness constraint, with the \textit{SmoothFDR} framework \citep{tansey2016falsediscoveryratesmoothing} represents as leading example. SmoothFDR utilize an Empirical Bayes two-group mixture model where the prior null probability $\pi_0$ is constrained by a Total Variation (TV) penalty over a graph structure, enforcing piecewise-constant solutions that effectively cluster significant hypotheses into discrete spatial blocks. The optimization is performed via an Expectation-Maximization (EM) algorithm, where the M-step requires solving a non-differentiable graph-weighted fused lasso problem, necessitates the use of the Alternating Direction Method of Multipliers (ADMM), involving repeated solution of linear systems of the form $(I + \rho L)x = b$, where $L$ is the graph Laplacian. While this approach provides sharp boundaries on discrete topologies, the complexity, scaling with the number of edges $E$ and the iterations required for ADMM convergence, imposes substantial computational costs. These costs are particularly pronounced when continuous domains require high-resolution discretization to avoid "staircase" artifacts, or as graph density increases.

\textbf{Positioning.} While definitions may overlap, our framework is most closely aligned with the regularization paradigm, characterized by two defining features: (1) we directly model the spatially-varying prior probability $\alpha(\text{loc})$ within the lFDR two-group mixture model, rather than reweighting p-values; and (2) the spatial dependency is encoded explicitly through smoothness regularization in an RKHS.

\section{Problem Setup}
\label{sec:problem_setup}
\subsection{Problem Formulation and the Spatially Varying Mixture Model}
Let the Hypothesis Index Space (HIS) be a continuous domain $\mathcal{H} \subset \mathbb{R}^D$. Our goal is to estimate a spatially varying local False Discovery Rate (lfdr) over $\mathcal{H}$. Let $\mathcal{P}$ denote the set of candidate data-generating distributions. For each location $\text{loc} \in \mathcal{H}$, we define a null hypothesis $H_{\text{loc}} \subset \mathcal{P}$. For any true distribution $P \in \mathcal{P}$, the set of true nulls is the subset of locations where the null hypothesis holds (assumed measurable):     $\mathcal{H}_0(P) := \{ \text{loc} \in \mathcal{H} \mid P \in H_{\text{loc}}\}.$
We assume our observed data, consisting of p-values $\{p_i\}_{i=1}^N$ at discrete spatial coordinates $\{\text{loc}_i\}_{i=1}^N \subset \mathcal{H}$, represent a finite sample from an underlying continuous p-value process, denoted $(p_{\text{loc}}(X))_{\text{loc} \in \mathcal{H}}$. Following \cite{Blanchard_2014}, this \textbf{conceptual} process must satisfy the \emph{joint measurability} condition, which states that the mapping $(\omega, \text{loc}) \mapsto p_{\text{loc}}(X(\omega))$ is jointly measurable. Notably, our method is designed for the case of arbitrary relations and does not require stronger assumptions such as Positive Regression Dependence on a Subset (PRDS). Instead, the structure is modeled implicitly through the choice of a reproducing kernel and smoothness regularization. Following \cite{Efron2004}, we model the p-value process by the conditional PDF where $\alpha(\text{loc})$ is the \emph{spatially varying prior probability} of $H_{\text{loc}}$:
\begin{equation}
    f(p \mid \text{loc}) = \alpha(\text{loc}) f_0(p) + (1 - \alpha(\text{loc})) f_1(p),
    \label{eq:mixture_model}
\end{equation}
Because $\alpha(\text{loc})$ represents a mixing probability, a critical constraint is that $\alpha(\text{loc}) \in [0, 1]$ for all $\text{loc} \in \mathcal{H}$. Here, $f_0(p)$ is the PDF for p-values drawn from a location where the null hypothesis is true. While validity assumptions typically require $f_0(p)$ to be uniform, we follow \cite{Efron2004} in allowing for more flexible, empirically estimated null distributions. Conversely, $f_1(p)$ is the PDF for p-values drawn from locations where \textbf{alternative hypotheses} are true. 

\subsection{Estimating the Component Densities $f_0$ and $f_1$}
\label{subsec:marginal_model}

A key advantage of this formulation is the separation between spatial structure and component densities (Proposition~\ref{prop:marginal}). Spatial dependence enters only through the mixing proportion $\alpha(\text{loc})$, while the null and alternative densities $f_0(z), f_1(z)$ remain spatially invariant.
Consequently, the marginal distribution of test statistics, pooled across all locations, follows the classical lFDR. Hence, we can estimate $f_0$ and $f_1$ using any methods for non-spatial lFDR (e.g central matching or empirical null fitting \cite{Efron2004}), treating spatial coordinates as irrelevant to the marginal estimation.

\begin{proposition}[Marginal Density Independence]
\label{prop:marginal}
Under the spatially-varying mixture model (Eq.~\ref{eq:mixture_model}), the marginal density of test statistics follows the standard two-group mixture:
\begin{equation}
    f(z) = \bar{\alpha} f_0(z) + (1-\bar{\alpha}) f_1(z)
\end{equation}
where $\bar{\alpha} = \mathbb{E}_{\text{loc}}[\alpha(\text{loc})]$ is the spatial average of the null probability.
\end{proposition}
Proof in Supplementary Section B. The specific estimation procedure employed in Section~\ref{sec:experiments}.

\section{Point-wise Solution}
\label{sec:pointwise}
We frame our problem as a Tikhonov-regularized maximum likelihood problem within a Reproducing Kernel Hilbert Space (RKHS), $\mathcal{H}_K$. Let $\|\cdot\|_{\mathcal{H}_K}$ denote the norm induced by the positive definite kernel $K: \mathcal{X} \times \mathcal{X} \to \mathbb{R}$. Our goal is to estimate $\alpha(\text{loc}_i)$ at the observed locations $\{\text{loc}_i\}_{i=1}^N$ where we have test statistics $\{z_i\}_{i=1}^N$ (or equivalently, p-values $\{p_i\}_{i=1}^N$). Notice Section \ref{sec:entire_domain} will address the case of estimating on the entire domain. 
We regularize toward the global prior $\bar\alpha$ rather than toward zero, so hypotheses with weak spatial information default to the marginal estimate and the $\lambda_{\text{reg}}\to\infty$ limit recovers the classical non-spatial two-group mixture. Let $\mathbf{c}_{\bar\alpha}$ denote the precomputed coefficient vector satisfying $K\mathbf{c}_{\bar\alpha} = \bar\alpha\,\mathbf{1}$,
obtained once via conjugate gradient before optimization begins; the centering adds no per-iteration cost and preserves the kernel-cancellation property established below. Our objective is then to minimize the negative log-likelihood:
\begin{align}
\min_{\alpha \in \mathcal{H}_K} \quad \mathcal{J}(\alpha)
 = -\sum_{i=1}^N \log\bigl( \alpha(\text{loc}_i) f_0(p_i)
   + (1-\alpha(\text{loc}_i)) f_1(p_i)\bigr)
   + \lambda_{\text{reg}} \|\alpha - \bar\alpha\|^2_{\mathcal{H}_K}
\label{eq:objective_function}
\end{align}
with $\lambda_{\text{reg}} > 0$ controlling the regularization strength. Since the empirical risk term depends on $\alpha$ solely through its evaluations at the finite set of points $\{\text{loc}_i\}_{i=1}^N$, the Generalized Representer Theorem \citep{Schlkopf2001} guarantees that the minimizer lies in the finite-dimensional subspace spanned by the kernel sections centered at the data: $\alpha(\cdot) = \sum_{i=1}^N c_i K(\cdot, \text{loc}_i).$ This reduces the variational problem to optimization over the coefficient vector $\mathbf{c} \in \mathbb{R}^N$. The core challenge in this formulation is restricting $\alpha(\text{loc}_i) \in [0,1]$ at all observed locations while maintaining a tractable convex optimization problem. We address this challenge here under the point-wise setting, with Section \ref{sec:entire_domain} extending the approach to the entire domain. Rather than imposing hard box constraints $\alpha \in [0,1]$, which would necessitate constrained optimization methods (see below), we introduce a \emph{soft boundary penalty} that penalizes violations per unit interval:
\begin{equation}
    \Lambda_{\text{bound}}(\alpha) = \sum_{i=1}^N \left[ \max(0, \alpha_i - 1)^2 + \max(0, -\alpha_i)^2 \right]
\end{equation}
This formulation maintains convexity (each term is a composition of the convex squared hinge loss with linear function values), is differentiable almost everywhere enabling efficient gradient-based optimization, and is inactive when constraints are satisfied ($\alpha_i \in [0,1]$), ensuring it does not interfere with well-behaved solutions. Alternative approaches include hard projection methods or squashing functions (e.g., logistic transformations); however, as detailed in Supplementary Section \ref{supp:alternative_methods}, we found the penalty-based formulation to be both elegant and sufficient in practice. Figure \ref{fig:full_width_estimation}(a) validates that $\alpha$ remains well-bounded within $[0,1]$ across all datasets evaluated in Section \ref{sec:experiments}, with \ref{fig:full_width_estimation}(b) shows this happens before convergence.
Our complete point-wise objective becomes then:
\begin{equation}
    \min_{\alpha \in \mathcal{H}_K} \mathcal{L}(\alpha)
    = -\sum_{i=1}^N \log h_i(\alpha)
      + \lambda_{\text{reg}} \|\alpha - \bar\alpha\|^2_{\mathcal{H}_K}
      + \lambda_{\text{bound}} \Lambda_{\text{bound}}(\alpha)
    \label{eq:pointwise_objective}
\end{equation}
where $h_i(\alpha) = \alpha(\text{loc}_i) f_0(p_i) + (1 - \alpha(\text{loc}_i)) f_1(p_i)$.

\begin{figure*}[t] 
    \centering
    \includegraphics[width=\textwidth]{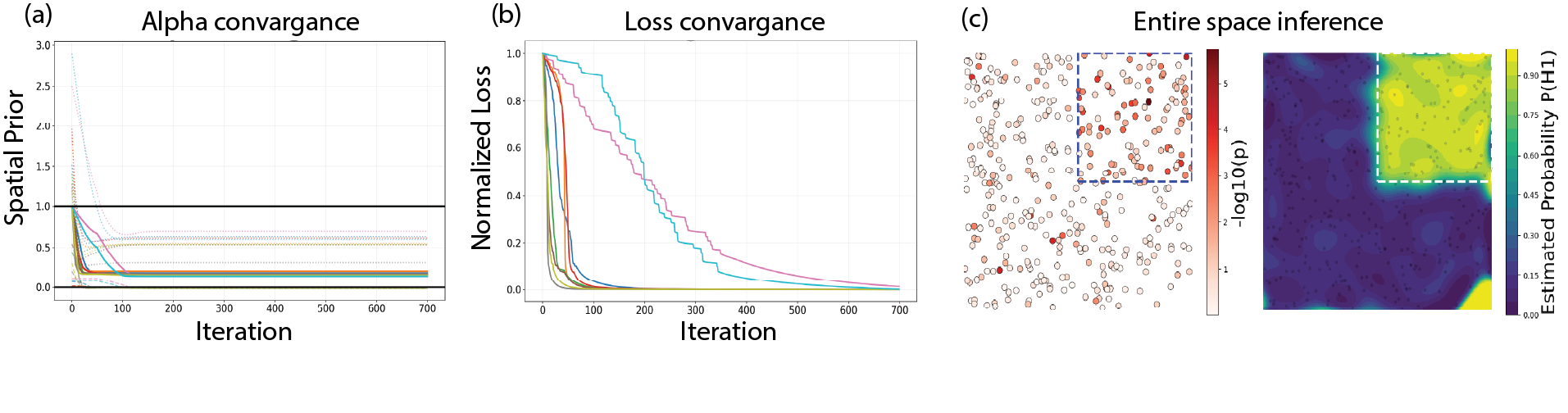}
    \caption{(a) Alpha convergence to [0,1] across all datasets- broken lines indicates min/max values, full line the mean. (b) Loss convergence. (c) Learning on entire space- Left: observed p-values. Right: $\alpha$ over the entire space.}
    \label{fig:full_width_estimation}
\end{figure*}

\paragraph{Optimization via Natural Gradient Descent}
While convex in $\mathbf{c} \in \mathbb{R}^N$, standard optimizers operating in Euclidean geometry struggle with RKHS objectives: for smooth kernels the Gram matrix $K$ has rapidly decaying eigenvalues, creating ill-conditioned optimization landscapes with narrow valleys. In our experiments, even sophisticated adaptive methods fail to overcome this. A more principled approach is the \emph{natural gradient} \citep{Munkres2018}, which accounts for the geometry of the parameter space, and proved beneficial in similar setups \cite{pmlr-v235-kozdoba24a}. For RKHS optimization, the natural gradient is defined with respect to the metric tensor $K$, $\tilde{\nabla}_{\alpha} \mathcal{L} = K^{-1} \nabla_{\mathbf{c}} \mathcal{L}$. Naively, this appears to trade one ill-conditioned operation for another, since $K^{-1}$ inherits the conditioning of $K$. Fortunately, the structure of our objective leads to an exact cancellation:

\begin{proposition}[Kernel Cancellation in Natural Gradient]
\label{lem:kernel_cancel}
For the point-wise objective in Equation~\ref{eq:pointwise_objective}, the natural gradient admits the simple form:
\begin{equation}
    \tilde{\nabla}_{\alpha} \mathcal{L}
    = \mathbf{w} + 2\lambda_{\text{reg}}\, (\mathbf{c} - \mathbf{c}_{\bar\alpha})
      + \lambda_{\text{bound}} \nabla_{\alpha} \Lambda_{\text{bound}},
    \quad
    w_i = -\frac{f_0(p_i) - f_1(p_i)}{\alpha_i f_0(p_i) + (1-\alpha_i) f_1(p_i)}
    \label{eq:natural_gradient_pointwise}
\end{equation}

With $[\nabla_{\alpha} \Lambda_{\text{bound}}]_i = 2(\alpha_i - 1)$ if $\alpha_i > 1$, $2\alpha_i$ if $\alpha_i < 0$, and $0$ otherwise. And $\mathbf{c}_{\bar\alpha}$ precomputed once.
\end{proposition}
Proof in Supplementary Section C. The update rule depends only on residuals $\mathbf{w}$ and current parameters $\mathbf{c}$, eliminating $K$ from the gradient step entirely. The kernel matrix is required only for the forward evaluation $\boldsymbol{\alpha} = K\mathbf{c}$, which is numerically stable and yields significantly faster convergence.

\section{Decision Rules and FDR Control}
\label{sec:decision_rules}
Here we address the translation of estimated spatial null probabilities $\hat{\alpha}(\text{loc})$ into formal rejection rules with provable FDR guarantees. The dominant paradigm \citep{lei2018adapt, Barber_2015, barber2020robust} achieves FDR guarantees by assuming calibrated p-values, reflecting a different philosophical stance than the one outlined in the introduction, where p-values may be erroneous and side information is a source of evidence no less important. Nonetheless, formal guarantees under the calibration assumption remain important. We present two complementary rules, both provide FDR control regardless of the spatial information validity. We also show that under correct spatial specifications, one strictly dominates the other.

\paragraph{The Model-Free FDR Control}
decision rule exploits Proposition~\ref{prop:marginal}, which establishes that the marginal distribution of p-values follows the classical 
two-group mixture regardless of spatial structure. Within the spatially 
enriched candidate set, the classical two-group guarantee of \cite{san-cai2007} applies directly. Crucially, given \cite{san-cai2007} is an average based rule, the reduced null proportion may yield discoveries beyond those of the purely marginal procedure. The procedure is as follows: Estimate $\hat{\alpha}(\text{loc}_i)$ via Stage~1 (Section~\ref{sec:pointwise}), and separately estimate $f_0, f_1, \bar{\alpha}$ from the marginal two-group mixture (Proposition~\ref{prop:marginal}). First, compute the marginal lfdr for every hypothesis: $\text{lfdr}_{\text{marg}}(p_i) = \frac{\bar{\alpha}\, f_0(p_i)}{\bar{\alpha}\, f_0(p_i) + (1 - \bar{\alpha})\, f_1(p_i)}.$ Then form the candidate set $S = \{i : \hat{\alpha}(\text{loc}_i) \leq q\}$,  and within $S$, reject the largest set $R_1 \subseteq S$ satisfying $\frac{1}{|R_1|}\sum_{i \in R_1} \text{lfdr}_{\text{marg}}(p_i) \leq q$.

\paragraph{The Spatial FDR Control via Mirror Statistics decision rule}
obtains FDR control through the mirror-statistic framework \citep{Barber_2015}. Define the \emph{rejection score} and \emph{mirror score}:
\begin{equation}
\begin{aligned}
    T_i &= \frac{\hat{\alpha}(\text{loc}_i)\, f_0(p_i)}{\hat{\alpha}(\text{loc}_i)\, f_0(p_i) + (1 - \hat{\alpha}(\text{loc}_i))\, f_1(p_i)}, 
    \tilde{T}_i = \frac{\hat{\alpha}(\text{loc}_i)\, f_0(1-p_i)}{\hat{\alpha}(\text{loc}_i)\, f_0(1-p_i) + (1 - \hat{\alpha}(\text{loc}_i))\, f_1(1-p_i)}
\end{aligned}
\label{eq:rejection_mirror_scores}
\end{equation}
The mirror score replaces $p_i$ with its complement $1-p_i$. The procedure 
computes $T_i$ and $\tilde{T}_i$ for all hypotheses, then sets the threshold 
via the Barber--Cand\`{e}s step-up rule:
\begin{equation}
    \hat{t}_\alpha = \max\left\{t :\; \frac{1 + \#\{i : \tilde{T}_i \leq t\}}{1 \vee \#\{i : T_i \leq t\}} \leq \alpha \right\}, \quad \mathcal{R} = \{i : T_i \leq \hat{t}_\alpha\}.
    \label{eq:threshold_rule}
\end{equation}
Under the spatially varying mixture model (Eq.~\ref{eq:mixture_model}) with mirror-conservative null p-values and $\hat{\alpha}$ obtained from the RKHS penalized likelihood with $\lambda_{\text{reg}} > 0$: the procedure controls FDR at level $\alpha + O(1/N)$; see Theorem~\ref{thm:rule2_fdr}. The $O(1/N)$ slack can be eliminated at the cost of $N$ leave-one-out estimator fits, but in our experiments this refinement was not empirically necessary. The argument adapts the mirror-statistic framework of \citet{barber2020robust} via an approximate invariance argument: the RKHS regularization ensures that flipping any single null's p-value to $1-p_i$ perturbs $\hat{\alpha}$ by at most $O(1/N)$, preserving the coin-flip symmetry required by the Barber--Cand\`{e}s inequality.

Having established that both rules provide valid FDR control, Theorem \ref{thm:spatial_dominance} shows that the Spatial rule is strictly more powerful under \textbf{correct structure specification}. Throughout, we use the marginal versions of FDR and FNR introduced by \citet{san-cai2007}: $\mathrm{mFDR} = \mathbb{E}[V]/\mathbb{E}[R]$ and $\mathrm{mFNR} = \mathbb{E}[T]/\mathbb{E}[N - R]$, where $V$ counts false rejections, $R$ counts total rejections, and $T$ counts missed true alternatives. These large-sample analogs of FDR/FNR are the natural objects for compound-decision optimality. The proof (Supplementary Section~\ref{supp:power_proof}) applies the compound decision framework of \citet{san-cai2007}: under correct specification, the spatial lfdr equals the oracle posterior $P(\theta_i = 0 \mid p_i, \text{loc}_i)$, which is the Bayes-optimal test statistic \citep{san-cai2007}.

\section{Solution Over the Entire Domain}
\label{sec:entire_domain}
While the point-wise approach yields reliable estimates at observed coordinates, extending inference to the entire continuous domain $\mathcal{H}$ requires ensuring that $\alpha(\text{loc}) \in [0,1]$ holds globally. This introduces a significant challenge: enforcing the bounds $\forall \text{loc} \in \mathcal{H}$ generates an infinite number of constraints, which is not possible even with the finite-dimensional coefficient vector. Formally, this is known as a \emph{Semi-Infinite Programming (SIP)} problem \citep{Hettich1993, Lpez2007}. In Supplementary Section \ref{supp:alternative_methods}, we provide a detailed analysis of classical SIP solvers (such as exchange methods and barrier functions) and demonstrate why they are ill-suited. 
Consequently, we favor a robust two-stage approach that separates the problem into two stages, each addressing a specific aspect of the challenge: \textbf{The first stage} is the described point-wise estimation which obtain estimates $\{\hat{\alpha}_i\}_{i=1}^N$ at the observed locations. \textbf{The second stage} treats these point-wise estimates as target labels and learns a globally valid function $\alpha(\text{loc}) \in [0,1]$ that closely approximates them across the entire domain. This is fundamentally a problem of \emph{learning a probability-valued function}: from an information-theoretic perspective, we seek a function $\alpha(\text{loc})$ that best approximates the target distribution $\hat{\alpha}(\text{loc})$ over the spatial domain. Leveraging the Kullback-Leibler (KL) divergence, and since each $\alpha(\text{loc})$ represents a Bernoulli parameter (the probability of the null hypothesis), the divergence between $\hat{\alpha}$ and $\alpha$ is: $D_{\text{KL}}(\hat{\alpha} \| \alpha) = \hat{\alpha}\log\frac{\hat{\alpha}}{\alpha} + (1-\hat{\alpha})\log\frac{1-\hat{\alpha}}{1-\alpha}.$ 

Conceptualizing the true spatial function $\alpha^*(\text{loc})$ as existing over the entire domain $\mathcal{H}$, with the point-wise estimates $\{\hat{\alpha}_i\}$ as noisy observations at discrete locations, we seek to minimize the expected KL divergence over the spatial domain plus a smoothness regularization:
\begin{equation}
    \min_{\alpha \in \mathcal{H}_K} \int_{\mathcal{H}} D_{\text{KL}}(\alpha^*(\text{loc}) \| \alpha(\text{loc})) \, p(\text{loc}) \, d\text{loc} + \lambda_{\text{global}} \|\alpha\|^2_{\mathcal{H}_K}
\end{equation}

where $p(\text{loc})$ measures over the domain and $\lambda_{\text{global}} > 0$ controls smoothness. 
For the data fidelity term, minimizing KL is equivalent to minimizing the cross-entropy, since $\hat{\alpha}$ entropy is constant with respect to $\alpha$. Hence, we can approximate the integral using our finite observations $\{\hat{\alpha}_i, \text{loc}_i\}_{i=1}^N$, treating $p(\text{loc})$ as an empirical measure concentrated at the data points, yielding the \emph{logistic loss}:
\begin{align}
    \mathcal{L}_{\text{CE}}(\alpha) = -\sum_{i=1}^N \left[\hat{\alpha}_i \log\alpha(\text{loc}_i) + (1-\hat{\alpha}_i) \log(1-\alpha(\text{loc}_i))\right]
    \label{eq:cross_entropy}
\end{align}
Crucially, the logistic loss is a \emph{proper scoring rule} \citep{Gneiting2007}, minimizing expectation when $\alpha = \hat{\alpha}$, ensuring that our learned function faithfully approximates the point-wise estimates without systematic bias.
To ensure $\alpha(\text{loc}) \in [0,1]$ everywhere while minimizing this discrepancy, we use a squashing function. The natural choice is the logistic function $\sigma(z) = \frac{1}{1+e^{-z}}$, which maps the entire real line to $[0,1]$. This leads us to parameterize $\alpha(\text{loc}) = \sigma(g(\text{loc}))$ where $g \in \mathcal{H}_K$ is unconstrained. The resulting optimization is then the \emph{kernel logistic regression}:
\begin{align}
    \mathbf{c}^* = \arg\min_{\mathbf{c} \in \mathbb{R}^N} -\sum_{i=1}^N \left[\hat{\alpha}_i \log\sigma(g_i) + (1-\hat{\alpha}_i) \log(1-\sigma(g_i))\right] + \lambda_{\text{global}} \mathbf{c}^T K \mathbf{c}
    \label{eq:kernel_logistic}
\end{align}
Depends on application, one may choose different regularization parameters $\lambda_{\text{reg}}$ (Stage 1) and $\lambda_{\text{global}}$ (Stage 2) to separately control the smoothness of the point-wise fit and the global interpolation. In practice, we typically set $\lambda_{\text{global}} \leq \lambda_{\text{reg}}$ to allow the global function to closely follow the point-wise estimates while maintaining spatial coherence. Another key practical consideration is how to initialize Stage 2. The coefficients $\mathbf{c}_{\text{point}}$ from Stage 1 provide a natural warm start, but correspond to a different parameterization (direct $\alpha$ vs. squashed $g$). We use the inverse logistic transformation $ \mathbf{c}^{(0)} = K^{-1} \cdot \text{logit}(\hat{\boldsymbol{\alpha}})$  where $\text{logit}(\alpha) = \log(\alpha/(1-\alpha))$ and we clip $\hat{\alpha}_i$ to $[\epsilon, 1-\epsilon]$ for $\epsilon > 0$ (e.g., $\epsilon = 0.01$) to avoid numerical issues. 
Crucially, the logistic loss formulation yields a convex optimization problem, avoiding the local minima issues inherent in direct squashing approaches, and by applying natural gradients as in Section \ref{sec:pointwise}, we obtain a remarkably simple update rule :
\begin{equation}
    \tilde{\nabla}_{\mathbf{c}} \mathcal{L}_{\text{logistic}} = (\boldsymbol{\sigma} - \hat{\boldsymbol{\alpha}}) + 2\lambda_{\text{global}} \mathbf{c}
    \label{eq:natural_gradient_logistic}
\end{equation}

Figure \ref{fig:full_width_estimation}(c) illustrates the alpha inference on a synthetic 2-dimensional dataset. The left panel shows the observed locations with p-values encoded by color, while the right panel presents the inferred alpha function. Notice that while the p-values are relatively noisy, our method was able to correctly infer the regions of high alpha. 

\paragraph{Efficient Testing via Optimal Experimental Design.}
The ability to infer $\alpha(\text{loc})$ at unseen locations enables a data-efficient testing protocol: strategically select locations to test, then leverage spatial structure for decisions at remaining locations. This raises two questions: (1) which locations should we test when designing the experiment? (2) how do we quantify uncertainty at untested locations to determine if predictions are reliable?
For Question 1, we select $N_0$ locations using density-aware A-optimal 
design (Supplementary Section~\ref{supp:greedy_algorithm}), which minimizes 
average prediction variance across the domain. A parameter $\gamma$ controls the cluster-isolation tradeoff: from dense regions ($\gamma > 0$) to isolated points ($\gamma < 0$), with $\gamma = 0$ yields uniform weighting.
For Question 2, we quantify uncertainty at untested locations via Laplace 
approximation (Supplementary Section~\ref{supp:posterior_variance_theorem}), 
which combines spatial uncertainty (distance from observed locations) and 
statistical uncertainty (Fisher information). We evaluate this framework empirically in Section~\ref{sec:experiments}.

\section{Practical Considerations}
\label{sec:practical}

\paragraph{Hyperparameter Selection.}
\label{sec:hyperparameter_selection}
A key advantage of our framework is principled likelihood-based cross-validation for all hyperparameters (regularization strength $\lambda_{\text{reg}}$, kernel type, length-scale, and smoothness). As Proposition~\ref{prop:normalization} shows, the conditional density $f(z \mid \text{loc}; \theta)$ is properly normalized for every $\theta$, since the spatial sampling distribution $p(\text{loc})$ does not depend on $\theta$, and the constraint $\int_{\mathcal{H}} \alpha(\text{loc}; \theta)\,p(\text{loc})\,d\text{loc} = \bar{\alpha}$ is fixed by the marginal data alone. 

\paragraph{Kernel Choice.}
\label{sec:kernel_selection}
Kernel choice encodes both the geometry and the smoothness of the function space. For continuous domains we employ Matérn kernels, whose smoothness parameter $\nu > d/2$ guarantees Sobolev regularity for reliable interpolation \citep{Khanfer2024, pmlr-v235-kozdoba24a}. For discrete topologies the framework extends seamlessly through graph kernels: diffusion kernels \citep{kondor2002diffusion, Smola2003, Tripathi2016} replace the $L_1$ Total Variation penalty of prior methods with an $L_2$ smoothness penalty $\sum_{(i,j)\in E} W_{ij}(\alpha_i - \alpha_j)^2$ that preserves network modularity (e.g., for gene--gene interactions), and hyperbolic embeddings via Sarkar's theorem \citep{Sarkar2012} accommodate hierarchical structures (e.g., gene ontologies) in just two dimensions, where Euclidean kernels would otherwise require dimension exponential in tree depth. Supplementary Section~\ref{supp:pick_kernel} provides the in-depth treatment for kernels selection.

\section{Evaluations}
\label{sec:experiments}

\begin{figure*}[t]
    \centering
    \includegraphics[width=\textwidth]{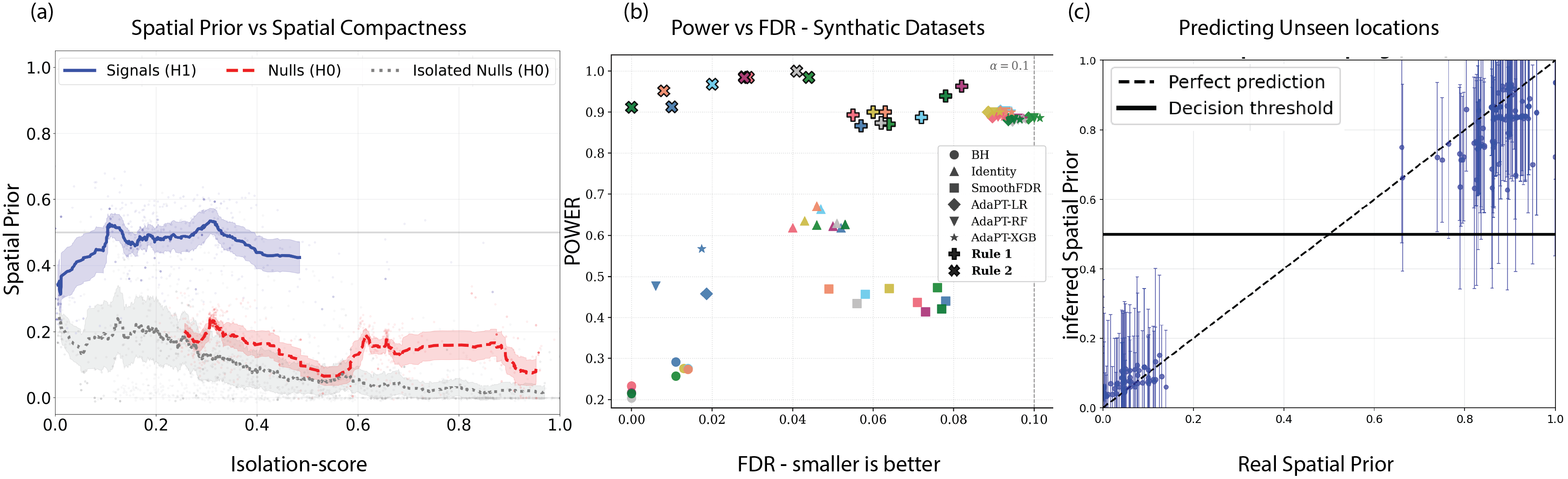}
    \caption{(a) Inferred spatial prior $(1-\alpha)$ versus geometric isolation score across semi-synthetic datasets. (b) Power vs. FDR on the 10 semi-synthetic datasets across all baselines at $\alpha = 0.10$. (c) Predicted spatial prior at held-out locations versus ground truth computed from full dataset. Error lines are the CI for 95\% confidence.}
    \label{fig:results}
\end{figure*}

Evaluating the proposed framework requires both known relations between hypotheses and ground-truth significance labels. We employ two complementary benchmarks: (i) a \textbf{semi-synthetic protocol} over 10 high-dimensional real-world datasets, where the spatial structure is real but p-values are injected with known labels; and (ii) two \textbf{real-world tasks} (HIGGS particle physics and TCGA differential expression) where both p-values and ground-truth labels are genuine. Across all evaluations, we employ both decision rules, which guarantees FDR control under arbitrary misspecification of the structure information. Baselines include Benjamini-Hochberg (BH), a non-spatial two-group model, SmoothFDR \citep{tansey2016falsediscoveryratesmoothing}, and AdaPT \citep{lei2018adapt} with three learners (logistic regression, random forest, XGBoost).

\paragraph{Semi-Synthetic Benchmarks.}
We utilize 10 high-dimensional real-world datasets onto which we synthetically inject p-values. For each dataset, we first cluster the data to identify stable clusters, then select two clusters mutually distant in kernel space: one serves as the signal-enriched cluster ($C_{\text{signal}}$), the other as the null-enriched cluster ($C_{\text{null}}$). Random null locations with no specific spatial structure are then added. A cluster corruption parameter ($\gamma=0.2$) randomly flips 20\% of the labels, creating outlier nulls and signals within otherwise homogeneous regions. P-values are generated from these labels, with nulls drawn from $\mathcal{U}[0, 1]$ and alternatives from $\text{Beta}(0.05, 5)$. Our method fits the local FDR using a Matérn kernel (distinct from the one used for clustering), with hyperparameters (length-scale, $\lambda_{\text{reg}}$) selected via cross-validation and boundary penalty fixed at $\lambda_{\text{bound}} = 500$. For the baselines, SmoothFDR converts p-values to z-scores via the probit transform and runs the graph fused lasso with a Gaussian two-group model; AdaPT uses 50 threshold iterations with the spatial coordinates as covariates.
Figure~\ref{fig:results}(a) plots $(1-\alpha)$ against geometric isolation, defined as inverted normalized kernel similarity to the $k$ nearest "$H_1$" neighbors, reflecting distance from signal-enriched regions. Red points belong to the $H_1$ cluster, blue to $H_0$, and gray to isolated locations. The method achieves clear separation between labels across isolation scores. While $H_0$ cluster maintain appropriate prior, they generally trend up because the null is distributed according to $U[0,1]$, so after spatial smoothing the low tail pulls the local mean, an effect less pronounced in isolated locations. 
Figure~\ref{fig:results}(b) reports FDR and power at $\alpha = 0.10$ across all datasets. While all methods control the FDR, \texttt{BH}, \texttt{Two-group}, and \texttt{SmoothFDR} forfeit substantial power (mean below $0.65$), so the meaningful comparison is against the AdaPT family, which saturates the FDR budget at mean power $0.843$--$0.855$. \textbf{Rule~1} matches AdaPT's power (mean $0.899$) while incurring roughly $20\%$ fewer false discoveries (mean FDR $0.066$ vs.\ $0.084$--$0.086$). \textbf{Rule~2} dominates the frontier outright: it attains strictly higher power than every baseline on every dataset, with gains over the best AdaPT variant ranging from $+0.030$ (\texttt{Hepatitis}) to $+0.346$ (\texttt{skin}), at a mean FDR of $0.023$, roughly a quarter of the available budget. The advantage is starkest on \texttt{skin}, where Rule~2 reaches power $0.913$ at FDR $0.010$ against AdaPT at $0.567$.

\paragraph{Real-World Datasets.}
\textbf{HIGGS particle physics}: the dataset \citep{Chen_2022} contains 98K simulated collision events binned into $N \approx 500$ groups by k-means clustering on kinematic features. For each bin, we compute a one-sided binomial p-value testing whether the rate of events satisfying $m_{wwbb} < 0.8$ exceeds the expected background rate. Ground-truth labels (signal-enriched vs. background-dominated bins) are used only for evaluation. \textbf{TCGA differential expression}: we test $N = 3{,}000$ genes between tumor types (BRCA cohort) using DESeq2 p-values from TCGA RNA-seq data. Ground-truth labels are derived from MSigDB Hallmark pathway membership: genes in 22 cancer-relevant pathways are labeled positive; genes in 4 unrelated pathways are labeled negative. FDR and power are evaluated on labeled discoveries. As for implementation details: For HIGGS, we employ a Matérn kernel, encoding nearby bins in kinematic space. For TCGA, we use a diffusion kernel ($\beta = 2$) on the STRING protein-protein interaction network \citep{Szklarczyk2022} (204,372 edges), widely used in systems biology to capture pathway-level signals \citep{Tripathi2016}; the diffusion parameter and $\lambda_{\text{reg}}$ are automatically tuned via C-V. For AdaPT baselines, covariates are bin centroids in kinematic space for HIGGS, and spectral coordinates of the PPI Laplacian for TCGA (matching the kernel information).

\begin{wraptable}{r}{0.5\columnwidth}
\centering
\caption{FDR\,/\,Power ($\alpha = 0.10$) .}
\label{tab:realworld}
\small
\setlength{\tabcolsep}{4pt}
\begin{tabular}{lcc}
\toprule
Method & HIGGS & TCGA DE \\
\midrule
BH        & .253/.564              & .065/.643              \\
Two-group  & .000/.000              & .065/.674              \\
SmoothFDR & .597/.922              & .540/.946              \\
AdaPT-LR  & .200/.059              & .082/.459              \\
AdaPT-RF  & .200/.078              & .085/.463              \\
AdaPT-XGB & .223/.426              & .075/.461              \\
\textbf{Rule~1} & .093/.382        & .061/.801               \\
\textbf{Rule~2} & .079/.343        & .063/.838              \\
\bottomrule
\end{tabular}
\end{wraptable}
As Table~\ref{tab:realworld} shows, on HIGGS, \texttt{BH} ($0.253$), \texttt{SmoothFDR} ($0.597$), and all three AdaPT variants ($\geq 0.20$) lose FDR control significantly, while \texttt{Two-group} achieves valid FDR only by making no rejections. \textbf{Rule~1} ($0.093$/$0.382$) and \textbf{Rule~2} ($0.079$/$0.343$) are the only methods that simultaneously control at $\alpha = 0.10$ and recover meaningful power. On TCGA DE, where most baselines do control FDR, the AdaPT variants stall at power $\approx 0.46$ and \texttt{BH}/\texttt{Two-group} plateau near $0.66$; Rule~1 and Rule~2 reach a $0.13$--$0.16$ power gain over the strongest FDR-controlling competitor at strictly tighter FDR.


\paragraph{Inference on Unseen Locations.}
To validate generalization to untested locations, we performed a holdout evaluation on the semi-synthetic datasets. We selected 70\% of locations using the density-aware A-optimal strategy ($\gamma=0.1$), observed their p-values, and fit the spatial model to predict $\hat{\alpha}(\text{loc})$ at the remaining 30\%. Figure~\ref{fig:results}(c) compares predictions (y-axis) against the true $\alpha$ values computed from the full dataset (x-axis), with 95\% posterior CI error bars. The strong diagonal clustering demonstrates accurate spatial interpolation, with 94.8\% coverage (close to nominal 95\%), enabling reliable decisions at untested locations.

\section{Conclusions and Limitations}
We introduced a unified framework for regularized FDR control on arbitrarily structured hypothesis spaces, where structure is encoded through reproducing kernel choice, and instantiated it via two decision rules, \emph{Rule~1} and \emph{Rule~2}. Our approach provides the first method to produce continuous prior maps valid across entire domains with convex optimization guarantees, flexible smoothness regularization and principled hyperparameter selection. Both rules deliver substantial power gains over existing baselines while maintaining robust FDR control and generalizing to unobserved locations.

\paragraph{Limitations.} While we believe our framework offers significant contributions, several limitations merit discussion. \textbf{First}, our theoretical guarantees apply to the finite set of observed locations; extending decisions to unobserved points relies on the continuous-domain procedure of Section~\ref{sec:entire_domain} and is subject to the standard caveats of out-of-sample inference. \textbf{Second}, every hypothesis contributes only a \emph{single p-value} from its marginal two-group mixture, creating an irreducible information floor that is most pronounced at isolated locations with few neighbors. This limitation applies to any method of this kind; an advantage of our formulation is that the resulting uncertainty can itself be quantified (Section~\ref{sec:entire_domain}). \textbf{Third}, the framework relies on kernel operations whose cost becomes a bottleneck for very large hypothesis sets; in Supplementary Section~\ref{sec:scalability} we demonstrate that the approach handles up to $500,000$ hypotheses on a single $24$\,GB GPU. Despite these limitations, we believe the substantial power gains demonstrated in our experiments, combined with the theoretical guarantees and computational efficiency of our framework, make it a valuable addition to the spatial FDR toolkit when structural assumptions are appropriate.

\newpage
\section{Societal Impact}
This paper presents work whose goal is to advance the field of machine learning. There are many potential societal consequences of our work, none of which we feel must be specifically highlighted here.

\bibliography{ref}
\bibliographystyle{abbrvnat}


\appendix
\onecolumn

\section{Proof of Marginal Density Independence}
\label{supp:marginal_proof}

\begin{proposition*}[Restated. Marginal Density Independence]
Under the spatially-varying mixture model (Eq.~\ref{eq:mixture_model}), the marginal density of test statistics follows the standard two-group mixture:
\begin{equation}
    f(z) = \bar{\alpha} f_0(z) + (1-\bar{\alpha}) f_1(z)
\end{equation}
where $\bar{\alpha} = \mathbb{E}_{\text{loc}}[\alpha(\text{loc})]$ is the spatial average of the null probability.
\end{proposition*}

\begin{proof}
Consider the observed locations $\{\text{loc}_i\}_{i=1}^N$ as samples drawn from a spatial sampling distribution with density $p(\text{loc})$ over the domain $\mathcal{H}$. The marginal density of a test statistic $z$, denoted $f(z)$, is obtained by integrating the conditional mixture model (Eq.~\ref{eq:mixture_model} in main text) over the spatial domain:

\begin{align}
    f(z) &= \int_{\mathcal{H}} f(z \mid \text{loc}) \, p(\text{loc}) \, d\text{loc} \\
    &= \int_{\mathcal{H}} \left[ \alpha(\text{loc}) f_0(z) + (1 - \alpha(\text{loc})) f_1(z) \right] p(\text{loc}) \, d\text{loc} \label{eq:supp_marginal_step1}
\end{align}

Since $f_0(z)$ and $f_1(z)$ do not depend on $\text{loc}$ (by our spatially-invariant assumption), we can factor them out of the integral:

\begin{align}
    f(z) &= f_0(z) \int_{\mathcal{H}} \alpha(\text{loc}) \, p(\text{loc}) \, d\text{loc} \nonumber \\
    &\quad + f_1(z) \int_{\mathcal{H}} (1 - \alpha(\text{loc})) \, p(\text{loc}) \, d\text{loc} \label{eq:supp_marginal_step2}
\end{align}

Define the global average null probability as:
\begin{equation}
    \bar{\alpha} = \mathbb{E}_{\text{loc}}[\alpha(\text{loc})] = \int_{\mathcal{H}} \alpha(\text{loc}) \, p(\text{loc}) \, d\text{loc}
    \label{eq:supp_alpha_bar}
\end{equation}

Then the second integral becomes:
\begin{align}
    \int_{\mathcal{H}} (1 - \alpha(\text{loc})) \, p(\text{loc}) \, d\text{loc}
    &= \int_{\mathcal{H}} p(\text{loc}) \, d\text{loc} - \int_{\mathcal{H}} \alpha(\text{loc}) \, p(\text{loc}) \, d\text{loc} \nonumber \\
    &= 1 - \bar{\alpha}
    \label{eq:supp_complement}
\end{align}

Substituting Equations~\eqref{eq:supp_alpha_bar} and \eqref{eq:supp_complement} into Equation~\eqref{eq:supp_marginal_step2}:
\begin{equation}
    f(z) = \bar{\alpha} f_0(z) + (1-\bar{\alpha}) f_1(z)
\end{equation}

This is precisely the standard two-group mixture form used in classical local FDR methods, confirming that the spatial structure affects only the mixing proportion $\bar{\alpha}$, not the functional forms of $f_0$ and $f_1$.
\end{proof}

\section{Proof of FDR Control for Rule 1}
\label{supp:rule1_proof}

\begin{proposition*}[Restated. FDR Control for Rule 1]
Let $S = \{i : \hat\alpha(\text{loc}_i) \le q\}$ be the spatial gate, assumed
independent of the p-values $\{p_i\}_{i\in S}$ used in the second stage of the
rule (achieved by sample-splitting; see remark below). Define
\[
\mathrm{lfdr}_{\mathrm{marg}}(p_i) \;=\; \frac{\bar\alpha\, f_0(p_i)}{\bar\alpha\, f_0(p_i) + (1-\bar\alpha)\, f_1(p_i)}.
\]
The rejection set
\[
R_1 \;=\; \arg\max\Bigl\{\, |R| \;:\; R\subseteq S,\;
       \tfrac{1}{|R|}\sum_{i\in R} \mathrm{lfdr}_{\mathrm{marg}}(p_i)\le q\,\Bigr\}
\]
controls the (Bayesian, marginal) FDR at level $q$ under the two-group mixture, regardless of the quality of $\hat\alpha$.
\end{proposition*}

\begin{proof}
The proof has two parts: validity of the average-lfdr step on a fixed candidate set, and preservation of validity under a p-value-independent spatial gate.

\paragraph{Part 1: Validity of average-lfdr thresholding on a fixed set.}
Fix any (deterministic or data-independent random) set $T\subseteq\{1,\dots,N\}$. Under the two-group mixture model, the Bayesian FDR \citep{Efron2004,san-cai2007} of any rejection set $R\subseteq T$ admits the closed-form representation
\begin{equation}
\label{eq:rule1_bfdr_identity}
    \mathbb{E}\!\left[\frac{|R\cap\mathcal{H}_0|}{|R|\vee 1}\,\Big|\,\{p_i\}_{i\in T}\right]
    \;=\; \frac{1}{|R|\vee 1}\sum_{i\in R}\mathrm{lfdr}_{\mathrm{marg}}(p_i),
\end{equation}
where $\mathrm{lfdr}_{\mathrm{marg}}(p_i) = P(H_0\mid p_i)$ under the marginal mixture. By Proposition~\ref{prop:marginal} (Marginal Density Independence), the marginal posterior null probability uses precisely the marginal $\bar\alpha$, $f_0$, $f_1$ that the procedure plugs in --- the spatial structure of $\alpha(\text{loc})$ does not enter $\mathrm{lfdr}_{\mathrm{marg}}$ at all.

The rule "select the largest $R\subseteq T$ such that $\frac{1}{|R|}\sum_{i\in R}\mathrm{lfdr}_{\mathrm{marg}}(p_i)\le q$" is the classical step-up rule of \citet{san-cai2007}, which controls $\mathrm{FDR}(R)\le q$ on $T$ as a direct consequence of identity~\eqref{eq:rule1_bfdr_identity}: the right-hand side is at most $q$ by construction, and it is equal in expectation to the FDP, so $\mathbb{E}[\mathrm{FDP}]\le q$.

\paragraph{Part 2: Spatial gate preserves validity.}
The candidate set $S = \{i : \hat\alpha(\text{loc}_i)\le q\}$ depends on $\hat\alpha$, which is fit on a sample-split portion of the data and is therefore independent of $\{p_i\}_{i\in S}$ used in the average-lfdr step. Conditioning on $S$ does not alter the calibration in~\eqref{eq:rule1_bfdr_identity}: for every $i\in S\cap\mathcal{H}_0$,
\begin{equation}
P(H_0\mid p_i,\, i\in S) \;=\; P(H_0\mid p_i) \;=\; \mathrm{lfdr}_{\mathrm{marg}}(p_i),
\end{equation}
because the event $\{i\in S\}$ is determined by $\hat\alpha(\text{loc}_i)$, which does not depend on $p_i$ in the sample-split design.

Applying Part 1 with $T=S$ yields $\mathrm{FDR}(R_1)\le q$.
\end{proof}

\begin{remark}[On the gate-independence assumption]
The independence between $S$ and $\{p_i\}_{i\in S}$ holds automatically in two natural designs:
(a) \textbf{Sample-splitting:} fit $\hat\alpha$ on a hold-out subset $\mathcal{D}_{\text{train}}$ disjoint from the testing set $\mathcal{D}_{\text{test}}$, and apply the gate only to indices $i\in\mathcal{D}_{\text{test}}$;
(b) $\hat\alpha$ uses only side-information covariates $x_i\neq p_i$ (no p-value information).
In settings without splitting, $\hat\alpha$ depends on the same p-values used for testing; the dependence is mediated by RKHS regularization with stability $\sup_x|\hat\alpha^{(i)}(x)-\hat\alpha(x)| = O(1/(N\lambda_{\text{reg}}))$ (\S\ref{supp:rule2_proof}, Step~2). The resulting inflation of the FDR bound is $O(1/N)$, analogous to the bound for Rule 2.
\end{remark}

\begin{remark}[Source of power gain]
Although the proof argues validity via subset-of-marginal containment, the rule is not merely conservative. Because $S$ is enriched for true alternatives whenever $\hat\alpha$ is informative, the empirical null proportion within $S$ is strictly smaller than $\bar\alpha$. The average-lfdr threshold $\frac{1}{|R|}\sum_{i\in R}\mathrm{lfdr}_{\mathrm{marg}}(p_i)\le q$ is therefore satisfied at a larger rejection budget than if applied to the full set $\{1,\dots,N\}$: more hypotheses can be admitted into $R$ while the running average remains below $q$. This is the source of Rule 1's power gain over the non-spatial procedure.
\end{remark}

\section{Derivation of Natural Gradient for Point-wise Optimization}
\label{supp:natural_gradient}

\subsection{Standard Euclidean Gradient}

The objective function (Eq.~\ref{eq:pointwise_objective} in main text) decomposes into three terms:
\begin{equation}
    \mathcal{L}(\mathbf{c}) = \mathcal{L}_{\text{data}}(\mathbf{c}) + \mathcal{L}_{\text{reg}}(\mathbf{c}) + \mathcal{L}_{\text{bound}}(\mathbf{c})
\end{equation}

We compute the gradient of each term separately.

\subsubsection{Data Term Gradient}

The negative log-likelihood term is:
\begin{equation}
    \mathcal{L}_{\text{data}} = -\sum_{i=1}^N \log\left[\alpha_i f_0(p_i) + (1-\alpha_i) f_1(p_i)\right]
\end{equation}

Using the chain rule with $\alpha_i = (K\mathbf{c})_i = \sum_{j=1}^N K_{ij} c_j$:

\begin{align}
    \frac{\partial \mathcal{L}_{\text{data}}}{\partial c_j}
    &= -\sum_{i=1}^N \frac{1}{\alpha_i f_0(p_i) + (1-\alpha_i) f_1(p_i)} \cdot \frac{\partial}{\partial c_j}\left[\alpha_i f_0(p_i) + (1-\alpha_i) f_1(p_i)\right] \\
    &= -\sum_{i=1}^N \frac{f_0(p_i) - f_1(p_i)}{\alpha_i f_0(p_i) + (1-\alpha_i) f_1(p_i)} \cdot \frac{\partial \alpha_i}{\partial c_j} \\
    &= -\sum_{i=1}^N \frac{f_0(p_i) - f_1(p_i)}{\alpha_i f_0(p_i) + (1-\alpha_i) f_1(p_i)} \cdot K_{ij}
\end{align}

In vector notation, define the residual vector $\mathbf{w} \in \mathbb{R}^N$ with entries:
\begin{equation}
    w_i = -\frac{f_0(p_i) - f_1(p_i)}{\alpha_i f_0(p_i) + (1-\alpha_i) f_1(p_i)}
\end{equation}

Then:
\begin{equation}
    \nabla_{\mathbf{c}} \mathcal{L}_{\text{data}} = K\mathbf{w}
    \label{eq:supp_grad_data}
\end{equation}

\subsubsection{Regularization Term Gradient (Centered)}

The centered RKHS regularization is:
\begin{equation}
    \mathcal{L}_{\text{reg}}
    \;=\; \lambda_{\text{reg}} \|\alpha - \bar\alpha\|^2_{\mathcal{H}_K}
    \;=\; \lambda_{\text{reg}} (\mathbf{c} - \mathbf{c}_{\bar\alpha})^T K (\mathbf{c} - \mathbf{c}_{\bar\alpha}),
\end{equation}
where $\mathbf{c}_{\bar\alpha}\in\mathbb{R}^N$ is the precomputed anchor coefficient defined by $K\mathbf{c}_{\bar\alpha} = \bar\alpha\mathbf{1}$ (solved once via conjugate gradient before optimization begins).

Taking the gradient:
\begin{equation}
    \nabla_{\mathbf{c}} \mathcal{L}_{\text{reg}}
    \;=\; 2\lambda_{\text{reg}} K (\mathbf{c} - \mathbf{c}_{\bar\alpha}).
    \label{eq:supp_grad_reg}
\end{equation}

\subsubsection{Boundary Penalty Gradient}

The boundary penalty is:
\begin{equation}
    \Lambda_{\text{bound}}(\alpha) = \sum_{i=1}^N \left[ \max(0, \alpha_i - 1)^2 + \max(0, -\alpha_i)^2 \right]
\end{equation}

Define the element-wise gradient:
\begin{equation}
    [\nabla_{\alpha} \Lambda_{\text{bound}}]_i = \frac{\partial}{\partial \alpha_i}\left[\max(0, \alpha_i - 1)^2 + \max(0, -\alpha_i)^2\right]
\end{equation}

This has the closed form:
\begin{equation}
    [\nabla_{\alpha} \Lambda_{\text{bound}}]_i = \begin{cases}
        2(\alpha_i - 1) & \text{if } \alpha_i > 1 \\
        2\alpha_i & \text{if } \alpha_i < 0 \\
        0 & \text{otherwise}
    \end{cases}
\end{equation}

Using the chain rule $\frac{\partial \Lambda_{\text{bound}}}{\partial c_j} = \sum_{i=1}^N \frac{\partial \Lambda_{\text{bound}}}{\partial \alpha_i} \frac{\partial \alpha_i}{\partial c_j}$:

\begin{equation}
    \nabla_{\mathbf{c}} \mathcal{L}_{\text{bound}} = \lambda_{\text{bound}} K \nabla_{\alpha} \Lambda_{\text{bound}}
    \label{eq:supp_grad_bound}
\end{equation}

\subsection{Combined Gradient and Factorization}

Combining Equations~\eqref{eq:supp_grad_data}, \eqref{eq:supp_grad_reg}, and \eqref{eq:supp_grad_bound}:

\begin{align}
    \nabla_{\mathbf{c}} \mathcal{L}
    &= K\mathbf{w} + 2\lambda_{\text{reg}} K (\mathbf{c} - \mathbf{c}_{\bar\alpha}) + \lambda_{\text{bound}} K \nabla_{\alpha} \Lambda_{\text{bound}} \\
    &= K \left[ \mathbf{w} + 2\lambda_{\text{reg}} (\mathbf{c} - \mathbf{c}_{\bar\alpha}) + \lambda_{\text{bound}} \nabla_{\alpha} \Lambda_{\text{bound}} \right]
    \label{eq:supp_factored_gradient}
\end{align}

This factorization reveals that the Gram matrix $K$ appears as a leading factor, which is the source of the ill-conditioning in standard gradient descent.

\subsection{Natural Gradient and Kernel Cancellation}

\begin{proof}[Proof of Lemma~\ref{lem:kernel_cancel}]
The natural gradient is defined as:
\begin{equation}
    \tilde{\nabla}_{\alpha} \mathcal{L} = K^{-1} \nabla_{\mathbf{c}} \mathcal{L}
\end{equation}

Substituting the factored form from Equation~\eqref{eq:supp_factored_gradient}:

\begin{align}
    \tilde{\nabla}_{\alpha} \mathcal{L}
    &= K^{-1} \left( K \left[ \mathbf{w} + 2\lambda_{\text{reg}} (\mathbf{c} - \mathbf{c}_{\bar\alpha}) + \lambda_{\text{bound}} \nabla_{\alpha} \Lambda_{\text{bound}} \right] \right) \\
    &= \left(K^{-1} K\right) \left[ \mathbf{w} + 2\lambda_{\text{reg}} (\mathbf{c} - \mathbf{c}_{\bar\alpha}) + \lambda_{\text{bound}} \nabla_{\alpha} \Lambda_{\text{bound}} \right] \\
    &= I \left[ \mathbf{w} + 2\lambda_{\text{reg}} (\mathbf{c} - \mathbf{c}_{\bar\alpha}) + \lambda_{\text{bound}} \nabla_{\alpha} \Lambda_{\text{bound}} \right] \\
    &= \mathbf{w} + 2\lambda_{\text{reg}} (\mathbf{c} - \mathbf{c}_{\bar\alpha}) + \lambda_{\text{bound}} \nabla_{\alpha} \Lambda_{\text{bound}}
\end{align}

where we used $K^{-1}K = I$ (the identity matrix). This completes the proof.
\end{proof}

The complete natural gradient descent update at iteration $k$ is:

\begin{algorithm}[H]
\caption{Natural Gradient Update for Point-wise FDR (centered)}
\begin{algorithmic}[1]
\State \textbf{Pre-compute (one time):} $\mathbf{c}_{\bar\alpha}$ by solving $K\mathbf{c}_{\bar\alpha} = \bar\alpha\mathbf{1}$ via CG. \Comment{$O(N^2 t_{\text{CG}})$}
\State \textbf{Input:} Current coefficients $\mathbf{c}^{(k)}$, learning rate $\eta$
\State Compute $\boldsymbol{\alpha}^{(k)} = K\mathbf{c}^{(k)}$ \Comment{Forward pass: $O(N^2)$}
\State Compute residuals: $w_i^{(k)} = -\frac{f_0(p_i) - f_1(p_i)}{\alpha_i^{(k)} f_0(p_i) + (1-\alpha_i^{(k)}) f_1(p_i)}$ \Comment{$O(N)$}
\State Compute boundary gradients: $[\nabla_{\alpha} \Lambda_{\text{bound}}]_i$ \Comment{$O(N)$}
\State Form natural gradient:
\State \quad $\tilde{\nabla}^{(k)} = \mathbf{w}^{(k)} + 2\lambda_{\text{reg}} (\mathbf{c}^{(k)} - \mathbf{c}_{\bar\alpha}) + \lambda_{\text{bound}} \nabla_{\alpha} \Lambda_{\text{bound}}$ \Comment{$O(N)$}
\State Update: $\mathbf{c}^{(k+1)} = \mathbf{c}^{(k)} - \eta \tilde{\nabla}^{(k)}$ \Comment{$O(N)$}
\State \textbf{Return:} $\mathbf{c}^{(k+1)}$
\end{algorithmic}
\end{algorithm}

The total per-iteration complexity is $O(N^2)$, dominated by the kernel matrix-vector product in the forward pass.

\section{Proof of Hyperparameter-Independent Normalization}
\label{supp:normalization_proof}

Here we provide the complete proof of Proposition~\ref{prop:normalization}.

\begin{proposition*}[Restated. Hyperparameter-Independent Normalization]
\label{prop:normalization}
The joint density $f(\text{loc}, z; \theta)$ integrates to 1 for all hyperparameter values $\theta$:
\begin{equation}
\int_{\mathcal{H}} \int_{z} f(\text{loc}, z; \theta) \, dz \, d\text{loc} = 1, \quad \forall \theta
\end{equation}
Consequently, the conditional density $f(z|\text{loc}; \theta)$ is a valid probability density for all $\theta$, enabling rigorous likelihood-based model selection.
\end{proposition*}

\begin{proof}[Proof of Proposition~\ref{prop:normalization}]

Let $\mathcal{H}$ denote the spatial domain and $[z_{\min}, z_{\max}]$ denote the support of the test statistics. We must show that the normalization constant:
\begin{equation}
Z(\theta) = \int_{\mathcal{H}} \int_{z_{\min}}^{z_{\max}} f(\text{loc}, z; \theta) \, dz \, d\text{loc}
\label{eq:supp_normalization_constant}
\end{equation}
equals 1 for all hyperparameter values $\theta$.

By the factorization $f(\text{loc}, z; \theta) = p(\text{loc}) \cdot f(z|\text{loc}; \theta)$ and the mixture model (Eq.~\ref{eq:mixture_model} in main text):
\begin{align}
Z(\theta) &= \int_{\mathcal{H}} \int_{z_{\min}}^{z_{\max}} p(\text{loc}) \cdot f(z|\text{loc}; \theta) \, dz \, d\text{loc} \nonumber \\
&= \int_{\mathcal{H}} p(\text{loc}) \int_{z_{\min}}^{z_{\max}} \left[ \alpha(\text{loc}; \theta)f_0(z) + (1-\alpha(\text{loc}; \theta))f_1(z) \right] dz \, d\text{loc}
\label{eq:supp_step1}
\end{align}

Since $f_0(z)$ and $f_1(z)$ do not depend on $\text{loc}$ (by assumption), we can factor them out of the inner integral:
\begin{align}
Z(\theta) &= \int_{\mathcal{H}} p(\text{loc}) \left[ \alpha(\text{loc}; \theta) \underbrace{\int_{z_{\min}}^{z_{\max}} f_0(z) \, dz}_{\text{Term A}} \right. \nonumber \\
&\qquad\qquad\qquad \left. + (1-\alpha(\text{loc}; \theta)) \underbrace{\int_{z_{\min}}^{z_{\max}} f_1(z) \, dz}_{\text{Term B}} \right] d\text{loc}
\label{eq:supp_step2}
\end{align}

Since $f_0$ and $f_1$ are valid probability density functions over $[z_{\min}, z_{\max}]$, by definition they integrate to 1:
\begin{align}
\text{Term A:} \quad &\int_{z_{\min}}^{z_{\max}} f_0(z) \, dz = 1 \\
\text{Term B:} \quad &\int_{z_{\min}}^{z_{\max}} f_1(z) \, dz = 1
\end{align}

Substituting into Equation~\eqref{eq:supp_step2}:
\begin{align}
Z(\theta) &= \int_{\mathcal{H}} p(\text{loc}) \left[ \alpha(\text{loc}; \theta) \cdot 1 + (1-\alpha(\text{loc}; \theta)) \cdot 1 \right] d\text{loc} \nonumber \\
&= \int_{\mathcal{H}} p(\text{loc}) \left[ \alpha(\text{loc}; \theta) + 1 - \alpha(\text{loc}; \theta) \right] d\text{loc} \nonumber \\
&= \int_{\mathcal{H}} p(\text{loc}) \, d\text{loc}
\label{eq:supp_step3}
\end{align}

Since $p(\text{loc})$ is itself a probability density function over the spatial domain $\mathcal{H}$, by definition:
\begin{equation}
\int_{\mathcal{H}} p(\text{loc}) \, d\text{loc} = 1
\label{eq:supp_step4}
\end{equation}

Finally, combining Equations~\eqref{eq:supp_step3} and \eqref{eq:supp_step4}:
\begin{equation}
Z(\theta) = 1, \quad \forall \theta
\end{equation}

Crucially, this derivation makes \textbf{no reference to the specific functional form of $\alpha(\text{loc}; \theta)$} beyond the requirement that $\alpha: \mathcal{H} \times \Theta \to [0,1]$.
This completes the proof.
\end{proof}

\subsection{Implications for Cross-Validation}

Proposition~\ref{prop:normalization} has several important consequences for hyperparameter selection:

\paragraph{Valid likelihood comparisons.}
Since $f(\text{loc}, z; \theta)$ is normalized for all $\theta$, the conditional density satisfies:
\begin{equation}
\int_{z} f(z|\text{loc}; \theta) \, dz = \frac{\int_{z} f(\text{loc}, z; \theta) \, dz}{p(\text{loc})} = \frac{p(\text{loc})}{p(\text{loc})} = 1
\end{equation}

This means test-set log-likelihoods $\sum_{i \in \text{Test}} \log f(z_i|\text{loc}_i; \theta)$ are directly comparable across different $\theta$ values without any normalization corrections.

The normalization constraint prevents pathological solutions where $\alpha(\text{loc}; \theta)$ is tuned to artificially inflate the marginal density $f(z)$ at the expense of spatial coherence. The mixture weights $\alpha(\text{loc}; \theta)$ must satisfy:
\begin{equation}
\int_{\mathcal{H}} \alpha(\text{loc}; \theta) \, p(\text{loc}) \, d\text{loc} = \bar{\alpha}
\end{equation}
where $\bar{\alpha}$ is determined by the marginal data, not by $\theta$.

\paragraph{Consistency with classical FDR.}
Because the regularizer is centered at $\bar\alpha$, the limit
$\lambda_{\text{reg}}\to\infty$ (infinite smoothing) drives
$\alpha(\text{loc};\theta)\equiv\bar\alpha$. In this limit the model reduces
exactly to the classical non-spatial two-group mixture, and the
cross-validation criterion selects the marginal null probability $\bar\alpha$
that maximizes the marginal likelihood, exactly as in standard local FDR
methods \citep{Efron2004}. (Without centering, the limit would instead drive
$\alpha\to 0$, breaking this consistency.)

\section{Proof of Asymptotic FDR Control for Rule 2}
\label{supp:rule2_proof}

\begin{theorem}[Asymptotic FDR Control; restated]
\label{thm:rule2_fdr}
Under the spatially varying mixture model (Eq.~\ref{eq:mixture_model}) with
mirror-conservative null p-values \citep{barber2020robust} (satisfied in
particular when $f_0 = 1$) and $\hat{\alpha}$ obtained from the centered RKHS
penalized likelihood with $\lambda_{\mathrm{reg}} > 0$:
\begin{equation}
    \mathrm{FDR}(\mathcal{R}) \leq \alpha + O\!\left(\tfrac{1}{N \lambda_{\mathrm{reg}}}\right).
\end{equation}
In particular, for any fixed $\lambda_{\mathrm{reg}} > 0$ the bound
reduces to $\alpha + O(1/N)$.
\end{theorem}

The proof proceeds in three steps. We first establish coin-flip symmetry
under the hypothetical scenario where $\hat{\alpha}$ does not depend on a
given null's p-value, then bound the deviation from this scenario via
stability of the regularized estimator, and finally propagate the bound
to the FDR.

\paragraph{Step 1: Coin-flip symmetry under exact invariance.}
Let $\mathcal{H}_0 \subseteq \{1, \dots, N\}$ denote the indices of true
nulls. By the mirror-conservative property, for each $i \in \mathcal{H}_0$
the joint distribution of the data is invariant under the swap
$p_i \leftrightarrow 1 - p_i$, conditional on $\{x_j\}_{j=1}^N$ and
$\{p_j\}_{j \neq i}$. If $\hat{\alpha}$ were a fixed function of inputs
that did not depend on $p_i$, then $T_i$ and $\tilde{T}_i$ would inherit
this exchangeability, since each is obtained by plugging $p_i$ or
$1-p_i$ into a common formula whose remaining arguments are unaffected
by the swap. This is the canonical setup of \citet{barber2020robust},
under which the step-up rule (Eq.~\ref{eq:threshold_rule}) controls
FDR exactly at level $\alpha$.

\paragraph{Step 2: Stability of $\hat{\alpha}$ under single-coordinate flips.}
The RKHS estimator solves
\begin{equation}
    \hat{\alpha} = \arg\min_{\alpha \in \mathcal{H}_K} \; -\frac{1}{N} \sum_{i=1}^N \log \ell_i\!\left(\alpha(x_i); p_i\right) + \lambda_{\text{reg}} \, \|\alpha - \bar{\alpha}\|_{\mathcal{H}_K}^2,
\end{equation}
where $\ell_i$ denotes the per-observation mixture likelihood. Let
$\hat{\alpha}^{(i)}$ denote the estimator obtained when $p_i$ is replaced
by $1 - p_i$ and all other observations are held fixed.

The objective is convex in $\alpha$ with modulus at least
$2\lambda_{\text{reg}}$ in the RKHS norm. Under the standing assumption
that $f_0$ and $f_1$ are bounded away from $0$ and $\infty$ on the
relevant range, the per-observation log-likelihood is uniformly bounded,
so changing a single observation perturbs the empirical loss term by
at most $C/N$. Standard stability arguments for regularized
M-estimators \citep{bousquet2002stability} then yield
\begin{equation}
    \|\hat{\alpha}^{(i)} - \hat{\alpha}\|_{\mathcal{H}_K} \leq \frac{C'}{N \lambda_{\text{reg}}},
\end{equation}
and consequently, by the reproducing property and bounded kernel,
\begin{equation}
    \sup_{x} \, \bigl|\hat{\alpha}^{(i)}(x) - \hat{\alpha}(x)\bigr| \leq \frac{C''}{N \lambda_{\text{reg}}}.
\end{equation}

\paragraph{Step 3: Propagation to FDR.}
The score map $(p, a) \mapsto a f_0(p) / [a f_0(p) + (1-a) f_1(p)]$ is
Lipschitz in $a$ on compact subsets, so there exists $L$ with
\begin{equation}
    \bigl|T_i - T_i^{(i)}\bigr|,\; \bigl|\tilde{T}_i - \tilde{T}_i^{(i)}\bigr| \leq \frac{L \cdot C''}{N \lambda_{\text{reg}}},
\end{equation}
where $T_i^{(i)}, \tilde{T}_i^{(i)}$ denote the scores constructed using
$\hat{\alpha}^{(i)}$ in place of $\hat{\alpha}$, i.e., the scores produced
by the hypothetical estimator that does not see $p_i$. By Step 1, the
collection $\{T_i^{(i)}, \tilde{T}_i^{(i)}\}_{i \in \mathcal{H}_0}$
satisfies exact coin-flip symmetry, so the Barber--Cand\`{e}s
supermartingale argument applied to these hypothetical scores delivers
exact FDR control at level $\alpha$. The actual scores deviate from
the hypothetical ones by $O(1/N \lambda_{\text{reg}})$ uniformly, which
inflates the resulting FDR bound by $O(1/N)$ when $\lambda_{\text{reg}}$
is held fixed. Combining:
\begin{equation}
    \mathrm{FDR}(\mathcal{R}) \leq \alpha + O(1/N).
\end{equation}
\hfill$\square$

\paragraph{Remark on misspecification.}
The argument relies only on the mirror-conservative property of null
p-values and on the stability of the regularized estimator. Consequently,
FDR control persists \emph{under arbitrarily wrong spatial information}:
a misspecified kernel, an irrelevant or adversarial covariate, or a
graph that does not reflect any true dependence cannot break the symmetry
of null p-values under the swap $p_i\leftrightarrow 1-p_i$. The role of
correct specification is to maximize \emph{power}, since under correct
specification $T_i$ approaches the true local fdr, which is the optimal
test statistic \citep{san-cai2007}.

\paragraph{Remark on the $\lambda_{\mathrm{reg}}$ dependence.}
The $1/(N\lambda_{\mathrm{reg}})$ slack arises from approximate, rather
than exact, leave-one-out invariance: stability of the regularized
estimator quantifies the perturbation in $\hat{\alpha}$ when a single
null's p-value is flipped, but the perturbation does not vanish.
Replacing $\hat{\alpha}$ with its leave-one-out counterpart
$\hat{\alpha}^{(i)}$ in the construction of $(T_i, \tilde{T}_i)$
recovers exact coin-flip symmetry under $H_0$ with $f_0 = 1$, removing
the $\lambda_{\mathrm{reg}}$ dependence and yielding
$\mathrm{FDR}(\mathcal{R}) \leq \alpha + O(1/N)$ unconditionally. The
cost is $N$ leave-one-out fits, which can be amortized via warm-starting
from $\hat{\alpha}$. We adopt the single-fit construction throughout for
computational simplicity and because, in the regimes considered, the
$\lambda_{\mathrm{reg}}$ chosen by likelihood cross-validation is bounded
away from zero.

\section{Proof of Power Dominance under Correct Specification}
\label{supp:power_proof}

\begin{theorem}[Optimality of Spatial lfdr]
\label{thm:spatial_dominance}
Assume correct specification, $\hat\alpha(\text{loc}) = \alpha^{*}(\text{loc})$, and that $\alpha^{*}(\text{loc})$ is non-constant. Then thresholding $\mathrm{lfdr}_{\mathrm{spatial}}(p_i,\text{loc}_i)$ achieves minimal $\mathrm{mFNR}$ among all decision rules controlling $\mathrm{mFDR}$ at level $\alpha$. In particular, thresholding $\mathrm{lfdr}_{\mathrm{marg}}(p_i)$ achieves strictly higher $\mathrm{mFNR}$ at the same level.
\end{theorem}

We adopt the compound decision framework of \citet{san-cai2007}, with $x_i = (p_i,\text{loc}_i)$ as the full observation for hypothesis $i$. The proof has five steps.

\paragraph{Step 1: Posterior under correct specification.}
Under the spatially varying mixture model, $\theta_i\sim\mathrm{Bernoulli}(1-\alpha^{*}(\text{loc}_i))$ independently and $p_i\mid\theta_i\sim f_{\theta_i}$, so the posterior null probability given the full observation is
\begin{equation}
    P(\theta_i = 0 \mid p_i,\text{loc}_i)
    \;=\; \frac{\alpha^{*}(\text{loc}_i)\, f_0(p_i)}
              {\alpha^{*}(\text{loc}_i)\, f_0(p_i) + (1-\alpha^{*}(\text{loc}_i))\, f_1(p_i)}
    \;=\; \mathrm{lfdr}_{\mathrm{spatial}}(p_i,\text{loc}_i).
\end{equation}

\paragraph{Step 2: Bayes-optimal decision rule.}
By Theorem~2 of \citet{san-cai2007}, the Bayes-optimal rule minimizing the weighted classification risk
\begin{equation}
    L_{\lambda}(\theta,\delta) \;=\; \frac{1}{N}\sum_{i=1}^N \bigl[\lambda\, I(\theta_i=0)\,\delta_i \;+\; I(\theta_i=1)(1-\delta_i)\bigr]
\end{equation}
is the test
\begin{equation}
    \delta_i^{*}
    \;=\; I\bigl[\,P(\theta_i=0\mid p_i,\text{loc}_i) < 1/\lambda\,\bigr]
    \;=\; I\bigl[\,\mathrm{lfdr}_{\mathrm{spatial}}(p_i,\text{loc}_i) < 1/\lambda\,\bigr].
\end{equation}
This minimizes $\mathbb{E}[L_{\lambda}]$ over \emph{all} decision rules (not just monotone or simple ones): the posterior risk decomposes as a sum over $i$, and each term is minimized pointwise by the Bayes test
\begin{equation}
    \delta_i^{*} \;=\; I\bigl[\,\lambda P(\theta_i=0\mid x_i) < P(\theta_i=1\mid x_i)\,\bigr],
\end{equation}
which reduces to the displayed form upon dividing both sides by $P(\theta_i=1\mid x_i)$.

\paragraph{Step 3: $\mathrm{mFDR}$/$\mathrm{mFNR}$ equivalence.}
By Theorem~1 of \citet{san-cai2007}, for any $\mathrm{mFDR}$ level $\alpha$ there exists a unique $\lambda^{*}(\alpha)$ such that $\delta^{*}(\lambda^{*}(\alpha))$ controls $\mathrm{mFDR}$ at exactly $\alpha$ with the smallest $\mathrm{mFNR}$ among all decision rules. This is established by showing that, along the family $\{\delta^{*}(\lambda)\}_{\lambda>0}$ of Bayes rules, $\mathrm{mFDR}$ is monotone increasing in $1/\lambda$ and continuous, so by the intermediate value theorem there is a unique threshold achieving any given $\mathrm{mFDR}\in(0,1)$. Admissibility of Bayes rules then gives that no other rule (Bayes or otherwise) can achieve strictly lower $\mathrm{mFNR}$ at the same $\mathrm{mFDR}$.

\paragraph{Step 4: Coarsening cost of the marginal rule.}
The marginal rule $\delta_i^{\mathrm{marg}} = I[\mathrm{lfdr}_{\mathrm{marg}}(p_i) < c]$ is a particular decision rule that ignores $\text{loc}_i$. For the value $c$ that yields $\mathrm{mFDR} = \alpha$, the corresponding $\mathrm{mFNR}$ satisfies
\begin{equation}
    \mathrm{mFNR}_{\mathrm{marg}} \;\ge\; \mathrm{mFNR}_{\mathrm{spatial}}
\end{equation}
by the optimality of $\delta^{*}$ established in Step 3.

\paragraph{Step 5: Strict inequality when $\alpha^{*}$ is non-constant.}
Equality in Step 4 holds only if $\mathrm{lfdr}_{\mathrm{spatial}}$ and $\mathrm{lfdr}_{\mathrm{marg}}$ induce the same ranking of hypotheses. Because $\mathrm{lfdr}_{\mathrm{marg}}$ depends only on $p_i$, it is constant on the level sets $\{j : p_j = p_i\}$, while $\mathrm{lfdr}_{\mathrm{spatial}}(p_i,\text{loc}_i)$ varies with $\text{loc}_i$ whenever $\alpha^{*}$ varies. Therefore the two rankings agree only if $\alpha^{*}(\text{loc})$ is constant in $\text{loc}$. Whenever $\alpha^{*}$ varies non-trivially, the inequality is strict. \hfill$\square$

\section{Marginal Density Estimation}
\label{supp:calc_fo_f1}

\paragraph{Motivation.}
The estimation of the marginal component densities, $f_0$ and $f_1$, typically follows one of two paradigms, each with distinct limitations. The first approach assumes the inputs are well-calibrated p-values, enforcing a strict Uniform distribution for the null hypothesis ($f_0 \sim U[0,1]$). While theoretically sound, this direction fails fundamentally when the input data are $z$-scores or other test statistics, and the calibration itself is rare in practice \cite{Efron2004}.

The second approach operates on $z$-values (or generic continuous data) and relies on general mixture models to separate the null and alternative distributions based on their shape. One might assume this direction could handle p-values by simply treating them as bounded data points. However, a nuanced problem arises: when the null hypothesis is truly Uniform, it manifests as a highly "overdispersed" or maximal-entropy background relative to the signal. Standard clustering or separation algorithms, which typically expect compact modes for both classes, fail to identify the signal against this flat, featureless background. This dilemma, where p-value methods cannot handle $z$-scores, and $z$-score methods fail on p-values due to null overdispersion, motivated many hybrid strategies, which we did not dive into here.

For the evaluations in this paper, we employ a simple hybrid strategy which respect the theoretical properties of the p-value domain to define the null, while utilizing the $z$-score domain to characterize the alternative signal.
First, we strictly enforce the theoretical null hypothesis, setting $f_0(p) = 1$ for all $p \in [0,1]$. This avoids the estimation instability caused by the overdispersed null. Second, to estimate the alternative density $f_1$, we transform the p-values into probit space ($z$-scores) via the inverse standard normal CDF, $z_i = \Phi^{-1}(1 - p_i)$. In this space, the alternative distribution is well-approximated by a Gaussian. We isolate the "signal" tail by selecting observations with $p_i < 0.2$ and estimate the alternative parameters $(\mu_1, \sigma_1)$ using the sample moments of these tail $z$-scores:
\begin{equation}
\hat{\mu}_1 = \text{mean}(z \mid p < 0.2), \quad \hat{\sigma}_1 = \text{std}(z \mid p < 0.2)
\end{equation}
The alternative density is then defined in $z$-space as $\mathcal{N}(\hat{\mu}_1, \hat{\sigma}_1)$ and transformed back to p-value space using the appropriate Jacobian:
\begin{equation}
f_1(p) = \frac{\phi(z_p \mid \hat{\mu}_1, \hat{\sigma}_1)}{\phi(z_p \mid 0, 1)}, \quad \text{where } z_p = \Phi^{-1}(1-p)
\end{equation}
This procedure yields a robust marginal model that combines the stability of the theoretical null with the flexibility of a parametric alternative fit.

\section{Kernel Selection Principles}
\label{supp:pick_kernel}

Having established both point-wise and entire-domain solution approaches, we now address the question of kernel choice. The reproducing kernel $K(\cdot, \cdot)$ serves a dual role: it defines the smoothness of $\alpha(x)$ through the RKHS norm $\|\alpha\|^2_{\mathcal{H}_K} = c^T K c$ (or, in centered form, $\|\alpha-\bar\alpha\|^2_{\mathcal{H}_K} = (\mathbf{c}-\mathbf{c}_{\bar\alpha})^T K (\mathbf{c}-\mathbf{c}_{\bar\alpha})$), and it determines the interpolation behavior between observed locations- essential for predictions at new points in the entire-domain setting. A fundamental requirement for any symmetric function $K: \mathbb{R}^d \times \mathbb{R}^d \to \mathbb{R}$ to serve as a reproducing kernel is positive definiteness. For our specific objective function with a logarithmic loss term $\log(\alpha(x_i))$, a slightly stronger condition is necessary. Since we employ natural gradient optimization, which requires computing $K^{-1}\nabla_{\mathbf{c}}\mathcal{L}$, the Gram matrix $K$ must be strictly positive definite ($K \succ 0$) and therefore invertible. This also ensures uniqueness of the solution in coefficient space, which is important given the logarithm's barrier-like behavior near zero. Most standard kernels (Gaussian, Matérn) satisfy this property when evaluated on distinct points.

While positive definiteness guarantees a valid RKHS, the \emph{differentiability} of the kernel determines the smoothness properties of functions in that space. This relationship is formalized through Sobolev space theory. A Sobolev space $W^{m,2}(\mathbb{R}^d)$ contains functions whose derivatives up to order $m$ are square-integrable, where $m$ quantifies smoothness. An RKHS can often be identified with a specific Sobolev space, with the kernel's differentiability determining $m$. A key result is the Sobolev embedding theorem is that if $m > d/2$, then functions in $W^{m,2}(\mathbb{R}^d)$ are guaranteed to be continuous and bounded.

The smoothness requirements depend on the application. If the goal is solely to estimate $\alpha(\text{loc}_i)$ at observed locations, kernel differentiability is not strictly required- the optimization machinery and Representer Theorem only require kernel evaluation at data points. However, when interpolating to new locations in the entire-domain setting, kernel differentiability becomes essential. By choosing a smooth kernel satisfying $m > d/2$, we impose the prior belief that $\alpha(x)$ varies continuously across space, ensuring the function behaves predictably between observed points. For example, the linear kernel $K(\mathbf{x}, \mathbf{y}) = \mathbf{x}^T \mathbf{y}$ is a valid choice for point-wise estimation, where it effectively uses the correlation between location vectors as a similarity measure. However, for entire-domain generalization, the linear kernel is constrained by its limited smoothness properties.
Given our need for a flexible, configurable kernel suitable for both point-wise and entire-domain settings, in our evaluations we employ the \textbf{Matérn kernel family} which includes a smoothness parameter $\nu > 0$ that explicitly controls the differentiability. To satisfy the Sobolev embedding condition for $d$-dimensional spaces, we require $\nu > d/2$.

\subsection{Graph Kernels: Bridging Discrete and Continuous FDR}

A significant conceptual advance of our framework is the connection it establishes between smooth FDR methods and graph kernel theory. While we focused on continuous spatial domains $\mathcal{H} \subset \mathbb{R}^d$, many multiple testing problems arise on discrete structures: gene regulatory networks, phylogenetic trees, gene ontology hierarchies, and social networks. Indeed, most current smooth FDR methods operate on such discrete graph representations, requiring topology-specific algorithmic solutions where methods designed for trees cannot be applied to scale-free networks, and techniques for hierarchical structures fail on general graphs.
Our RKHS framework provides a unified solution: by choosing an appropriate graph kernel, the same optimization algorithm applies to arbitrary discrete topologies.

For a comprehensive treatment of graph kernels, we refer the reader to the seminal work of Kondor and Lafferty (2002) and its subsequent extensions by Smola and Kondor (2003). Within this framework, the \textit{Graph Laplacian kernel} ($K = L^{\dagger}$) serves as a canonical example. Its associated RKHS norm, $\|\alpha\|^2_{\mathcal{H}} = \boldsymbol{\alpha}^T L \boldsymbol{\alpha} = \sum_{(i,j) \in E} W_{ij}(\alpha_i - \alpha_j)^2$, explicitly penalizes signal differences across edges using an $L_2$ metric. This formulation contrasts fundamentally with the Total Variation (TV) penalty ($\sum W_{ij}|\alpha_i - \alpha_j|$) employed by current graph-based FDR methods, whereas TV relies on the $L_1$ norm to enforce piecewise-constant clustering, the Laplacian kernel promotes smooth variation across the network. To capture longer-range dependencies, this approach naturally extends to \textit{diffusion kernels} ($K = \exp(-\beta L)$) and \textit{random walk kernels}, which incorporate global graph topology and allow distant but well-connected nodes to influence local estimation.
A particularly promising direction involves the use of \textbf{Hyperbolic Embeddings for Hierarchical Structures}, as many biological and social networks exhibit inherent hierarchical organization : gene ontologies (GO), phylogenetic trees, and protein interaction networks with hub-spoke patterns. For such structures, the seminal result of \cite{Sarkar2012} regarding hyperbolic spaces applies. Briefly, Sarkar's Theorem states that any tree with $n$ nodes can be embedded into the 2-dimensional hyperbolic space (Poincaré disk $\mathbb{H}^2$) with arbitrarily low distortion $(1+\epsilon)$ for any $\epsilon > 0$.This dimension efficiency is crucial for practical application, but more fundamentally, it addresses a geometric incompatibility. Consider gene ontology enrichment testing containing 50,000 terms, but observations at only 500 locations ($1\%$ sampling). In Euclidean space, such hierarchical structures face an intrinsic "capacity" problem: trees exhibit exponential volume growth, whereas Euclidean space has only polynomial growth. Consequently, embedding a tree into \emph{any} low-dimensional Euclidean space $\mathbb{R}^d$ must incur significant distortion. Even high-dimensional Euclidean embeddings fail to capture the tree topology faithfully compared to the hyperbolic plane. \textbf{In contrast}, by leveraging Sarkar's construction, we can embed these dependencies into just 2 dimensions with near-perfect fidelity, effectively bypassing the limitations of Euclidean kernels for hierarchical data.

\section{Alternative Approaches for Enforcing $\alpha(x) \in [0,1]$ Over Continuous Domains}
\label{supp:alternative_methods}

Before developing the convex two-step kernel logistic regression framework, we explored several direct methods for enforcing $\alpha(\text{loc}) \in [0,1]$ across the entire continuous domain $\mathcal{H}$. This section documents these approaches, their theoretical foundations, and the computational barriers that led us to pursue the convex alternative. We present this analysis to contextualize our methodological choices and guide future research.

\subsection{Approach 1: Direct Squashing and Non-Convexity}
A common solution for ensuring global constraint satisfaction is to apply a squashing function $\sigma: \mathbb{R} \to [0,1]$ to an unconstrained function, for example, the logistic function $\sigma(z) = 1/(1+e^{-z})$. The fundamental issue is that the composition of the logarithm with the mixture under squashing renders the objective non-convex. To see this explicitly, consider the data term for a single observation:
\begin{equation}
\ell(g_i) = -\log\left[\sigma(g_i) f_0(p_i) + (1-\sigma(g_i)) f_1(p_i)\right]
\label{eq:direct_squashing}
\end{equation}
The second derivative with respect to the latent value $g_i$ is:
\begin{align}
\frac{\partial^2 \ell}{\partial g_i^2} &= \frac{\partial}{\partial g_i}\left[\frac{-\sigma'(g_i)(f_0(p_i) - f_1(p_i))}{\sigma(g_i)f_0(p_i) + (1-\sigma(g_i))f_1(p_i)}\right] \\
&= \frac{-\sigma''(g_i)(f_0 - f_1)(\sigma f_0 + (1-\sigma)f_1) + (\sigma'(g_i))^2(f_0 - f_1)^2}{[\sigma(g_i)f_0(p_i) + (1-\sigma(g_i))f_1(p_i)]^2}
\end{align}

For the logistic function, the second derivative $\sigma''(z) = \sigma(z)(1-\sigma(z))(1-2\sigma(z))$ changes sign depending on whether $\sigma(z)$ is above or below $0.5$. This means the loss function is neither convex nor concave in $g_i$, even for a single data point. The global objective, being a sum over $N$ such terms plus a convex regularizer, inherits this non-convexity.

The non-convex landscape leads to multiple stationary points, many of which are local minima. Different initializations of $\mathbf{c}$ can converge to qualitatively different solutions with vastly different objective values. Moreover, the objective surface often exhibits regions where the gradient is near-zero but the Hessian has both positive and negative eigenvalues, causing optimization algorithms to stall. In our experiments, we found that the direct squashing approach failed to converge within reasonable iteration budgets, oscillating between different regions of parameter space. Even when optimization converged, the resulting $\alpha(\text{loc})$ functions often exhibited pathological behavior such as rapid oscillations between extreme values (near 0 and 1) in regions with sparse data.

\subsection{Approach 2: The Semi-Infinite Programming (SIP) Formulation}
Enforcing $\alpha(x) \in [0,1]$ for all $x \in \mathcal{H}$ is a \textbf{semi-infinite programming (SIP)} problem \citep{Hettich1993}, an optimization with finite decision variables but infinitely many constraints:
\begin{equation}
\begin{aligned}
\min_{\mathbf{c} \in \mathbb{R}^N} \quad & -\sum_{i=1}^N \log\left[\alpha(x_i) f_0(p_i) + (1-\alpha(x_i)) f_1(p_i)\right] + \lambda \|\alpha-\bar\alpha\|^2_{\mathcal{H}_K} \\
\text{s.t.} \quad & 0 \leq \alpha(x) \leq 1, \quad \forall x \in \mathcal{H}
\end{aligned}
\label{eq:sip}
\end{equation}
where $\alpha(x) = \sum_{j=1}^N c_j K(x, x_j)$.

Classical SIP methods include: (1) discretization with adaptive refinement, (2) exchange methods that iteratively add violated constraints, (3) reduction to finite equivalent constraints via problem structure, and (4) barrier/penalty methods with sampling-based approximation. Here, each method corresponds to a classical SIP approach: local reduction uses KKT finite reduction \citep{still2001generalized}, polynomial SDP applies moment relaxations \citep{henrion2009convex}, and the barrier method employs interior point penalties \citep{royset2013optimal}. However, SIP theory assumes convex objectives and our non-convex mixture likelihood eliminates convergence guarantees, and constraint violation patterns depend on kernel choice in complex ways.

\paragraph{Method 1: Local Reduction via Critical Points.} Extreme constraint violations can only occur at critical points in the interior or at domain boundaries. For upper boundary violations ($\alpha(x) > 1$), critical points satisfy:
\begin{align}
\alpha(x) &= \sum_{j=1}^N c_j K(x, x_j) = 1 \label{eq:boundary}\\
\nabla_x \alpha(x) &= \sum_{j=1}^N c_j \nabla_x K(x, x_j) = 0 \label{eq:gradient_constraint}\\
\nabla^2_{xx} \alpha(x) &\preceq 0 \label{eq:hessian_constraint}
\end{align}

For Gaussian RBF kernels, this yields a system of $d+1$ nonlinear equations in $d$ unknowns, solvable via Newton-Raphson when well-conditioned. By Bézout's theorem, a system of $d$ polynomial equations of degree $\mathcal{O}(d)$ admits $\mathcal{O}(2^d d^d)$ real solutions. Cost per critical point: $\mathcal{O}(d^3)$ per Newton iteration gives a total cost of:
\begin{equation}
\text{Cost}_{\text{local}} = \mathcal{O}(2^d d^{d+3} N)
\end{equation}
So, while under compactness of $\mathcal{H}$ and non-degeneracy of $K$, all constraint-violating critical points are identified in finite time, the solution is intractable for large d due to exponential scaling and ill-conditioning.
Moreover, Jacobian conditioning deteriorates as $\kappa(J) \sim \mathcal{O}(N^{d/(d+2)})$, causing numerical instability. In practice, we tried few common solvers, all weren't able to solve the set of equations fro $d \geq 10$.

\paragraph{Method 2: Polynomial SDP Relaxation.} Approximate the non-polynomial mixture likelihood via Taylor expansion, then apply semidefinite programming relaxations for global certificates. Expand $\log(h_i(\alpha))$ around $\alpha = 0.5$ with $\beta_i = \alpha(x_i) - 0.5$ and mixing ratio $r_i = (f_0(p_i) - f_1(p_i))/m_i$:
\begin{equation}
\log(h_i) \approx \log(m_i) + \beta_i r_i - \frac{(\beta_i r_i)^2}{2} + \frac{(\beta_i r_i)^3}{3} - \frac{(\beta_i r_i)^4}{4}
\end{equation}

Substituting $\beta_i = \sum_j c_j K(x_i, x_j) - 0.5$ yields a degree-4 polynomial in $\mathbf{c}$. The Lasserre SDP hierarchy constructs moment matrices $M_k(y)$ of dimension $\binom{N+k}{k}$ that provide increasingly tight convex relaxations.

For bounded densities with $\max_i |f_0(p_i)|, |f_1(p_i)| \leq B$, Taylor error is $\mathcal{O}(|\beta_i|^5)$. As relaxation order $k \to \infty$, SDP converges to the global polynomial optimum. Nevertheless, the computational memory requirements for moment matrix:
\begin{equation}
\text{Memory} = 8 \cdot \frac{\binom{N+k}{k}(\binom{N+k}{k}+1)}{2} \text{ bytes}
\end{equation}

Concrete limits:
\begin{align*}
N=50, k=2: &\quad 7.03 \text{ MB} \\
N=100, k=2: &\quad 106 \text{ MB} \\
N=50, k=4: &\quad 400 \text{ GB}
\end{align*}

Making this solution non-practical.

\subsection{Approach 3: Barrier Method with Tail-Aware Sampling}
Finally, here we describe the barrier method for which we compare against in the evaluations section.
Transform the semi-infinite constraint set into a penalized objective using logarithmic barrier functions, then approximate the resulting domain integral via importance sampling that concentrates samples in regions most likely to violate constraints. The classical barrier method for constrained optimization replaces hard constraints with smooth penalty terms that approach infinity at the boundary. For the semi-infinite programming problem~\eqref{eq:sip} we reformulate as:
\begin{equation}
\mathcal{L}_{\nu}(\mathbf{c}) = -\sum_{i=1}^N \log(h_i(\alpha(x_i))) + \lambda \|\alpha-\bar\alpha\|^2_{\mathcal{H}_K} - \frac{1}{\nu} \int_{\mathcal{H}} \left[\log(\alpha(x)) + \log(1-\alpha(x))\right] dx
\end{equation}
where $\nu > 0$ is the barrier parameter controlling penalty strength and the integral is over the Lebesgue measure on $\mathcal{H} \subseteq \mathbb{R}^d$. The barrier terms $-\log(\alpha(x))$ and $-\log(1-\alpha(x))$ create increasingly steep penalties as $\alpha(x)$ approaches 0 or 1, respectively. As $\nu \to \infty$, these penalties force $\alpha(x)$ to remain strictly within $(0,1)$ throughout the domain. The method proceeds by solving a sequence of problems with increasing $\nu$:
\begin{equation}
\nu_0 = 1, \quad \nu_{k+1} = \beta \nu_k, \quad \beta \in [1.1, 1.5]
\end{equation}
Starting from small $\nu_0$ (weak constraints) allows easier optimization, while gradually increasing $\nu$ tightens constraints. Each iteration warm-starts from the previous solution.
The barrier objective contains the integral:
\begin{equation}
I(\mathbf{c}) = \int_{\mathcal{H}} \left[\log(\alpha(x)) + \log(1-\alpha(x))\right] dx
\end{equation}
This integral has \textbf{no closed form for general kernels and domains}, with direct numerical quadrature (e.g., Gaussian quadrature grids) becomes intractable in large d.

\subsubsection{Monte Carlo Approximation via Importance Sampling.}
We approximate the integral using Monte Carlo integration with $M$ samples $\{z_m\}_{m=1}^M$ drawn from a proposal distribution $q(x)$:
\begin{equation}
I(\mathbf{c}) = \int_{\mathcal{H}} \frac{f(x)}{q(x)} q(x) \, dx \approx \frac{1}{M} \sum_{m=1}^M \frac{f(z_m)}{q(z_m)}
\end{equation}
where $f(x) = \log(\alpha(x)) + \log(1-\alpha(x))$ and $q(x)$ is the sampling density. The key question is: \textit{what distribution $q$ minimizes variance and ensures constraint violation detection?}
For that, we construct a kernel density estimate from the $N$ observed data locations:
\begin{equation}
\hat{p}(x) = \frac{1}{N} \sum_{i=1}^N K_{\text{KDE}}(x, x_i)
\end{equation}
where $K_{\text{KDE}}(x, y) = \frac{1}{h^d} K_0\left(\frac{x-y}{h}\right)$ is a probability kernel (typically Gaussian) with bandwidth parameter $h > 0$. The bandwidth $h$ critically determines sampling quality. We employ two standard methods:
\textit{1. Scott's Rule}:
\begin{equation}
h_{\text{Scott}} = \left(\frac{4}{d+2}\right)^{\frac{1}{d+4}} N^{-\frac{1}{d+4}} \hat{\sigma}
\end{equation}
where $\hat{\sigma}$ is the empirical standard deviation of the data locations (computed per-dimension and averaged).
\textit{2. Cross-Validation} (optimal, expensive):
\begin{equation}
h_{\text{CV}} = \arg\min_h \frac{1}{N} \sum_{i=1}^N \left(\hat{p}_{-i}(x_i; h) - \delta(x_i)\right)^2
\end{equation}
where $\hat{p}_{-i}(x; h)$ is the leave-one-out KDE excluding point $x_i$, and $\delta$ is the Dirac delta. In practice, this minimizes integrated squared error via grid search over candidate $h$ values.

\textbf{The Critical Observation is} that naive sampling from $\hat{p}(x)$ concentrates samples where data is dense. However, constraint violations $\alpha(x) \notin [0,1]$ may occur most violently in \textit{tail regions} far from training data, where RKHS extrapolation becomes unreliable. Sampling from $\hat{p}(x)$ thus \textit{misses exactly the regions we need to monitor}. To address this, we design a hybrid distribution that balances three competing objectives:
\begin{equation}
p_{\text{hybrid}}(x) = \rho_1 \hat{p}(x) + \rho_2 q_{\text{tail}}(x) + \rho_3 u(x)
\end{equation}
where:
\textbf{Data-Dense Zone} ($\rho_1 = 0.3$) are sampled from KDE $\hat{p}(x)$. These regions contribute most to the data-fit term in $\mathcal{L}_{\nu}$.
\textbf{Tail Zone} ($\rho_2 = 0.5$) are samples from inverse density:
\begin{equation}
q_{\text{tail}}(x) = \frac{w(x)}{\int_{\mathcal{H}} w(y) dy}, \quad w(x) = \frac{1}{\hat{p}(x) + \epsilon}
\end{equation}
where $\epsilon > 0$ (typically $\epsilon = 10^{-6}$) prevents division by zero. These are points where violent violations occur, for which we use rejection sampling with acceptance probability $\propto (\hat{p}(x) + \epsilon)^{-1}$.
\textbf{Uniform Zone} ($\rho_3 = 0.2$) are sampled uniformly:
\begin{equation}
u(x) = \frac{1}{\text{Vol}(\mathcal{H})}
\end{equation}
Finally, the choice $(\rho_1, \rho_2, \rho_3) = (0.3, 0.5, 0.2)$ comes from hyperparameter optimization.

\subsubsection{Practical Limitations}
Despite strong theoretical foundations and favorable asymptotic complexity, the barrier method proved impractical in our experiments:

\begin{enumerate}
\item \textbf{Sampling Sensitivity:} The mixing weights $(\rho_1, \rho_2, \rho_3)$ require problem-specific tuning. Too much tail allocation ($\rho_2 > 0.6$) introduces excessive gradient noise; too little ($\rho_2 < 0.3$) misses violations.

\item \textbf{Bandwidth Selection:} KDE bandwidth $h$ critically affects tail sampling quality. Scott's rule often oversmooths for $d \geq 5$; cross-validation is expensive.

\item \textbf{Barrier Schedule:} The growth rate $\beta$ and starting value $\nu_0$ require careful tuning. Aggressive schedules ($\beta > 1.5$) cause convergence failure; conservative schedules ($\beta < 1.1$) waste computation.

\item \textbf{Gradient Variance:} Even with $M = 10^4$ samples, Monte Carlo variance in the barrier gradient necessitates small learning rates ($\eta \sim 10^{-3}$), requiring hundreds of iterations per barrier level.

\item \textbf{Non-Convexity Persists:} The method still optimizes the non-convex mixture likelihood, inheriting all local minima issues. Multiple random initializations are required, multiplying computational cost.

\item \textbf{Numerical Instability:} Near constraint boundaries ($\alpha \approx 0$ or $\alpha \approx 1$), the terms $1/\alpha$ and $1/(1-\alpha)$ become numerically unstable, requiring careful regularization ($\alpha \gets \text{clip}(\alpha, 10^{-8}, 1-10^{-8})$).
\end{enumerate}

\section{Practical Cross-Validation Procedure}
\label{supp:cv_procedure}

In our experiments (Section~\ref{sec:experiments} in main text), we employ the following procedure:

\begin{algorithm}[H]
\caption{Hyperparameter Selection via Cross-Validation}
\label{alg:supp_cv}
\begin{algorithmic}[1]
\Require Data $\{(p_i, \text{loc}_i)\}_{i=1}^N$, hyperparameter grid $\Theta$, number of folds $K$
\Ensure Optimal hyperparameters $\theta^*$
\State Partition data into $K$ folds: $\mathcal{D} = \mathcal{D}_1 \cup \cdots \cup \mathcal{D}_K$
\For{each $\theta \in \Theta$}
    \State Initialize cumulative log-likelihood: $\mathcal{L}_{\text{CV}}(\theta) = 0$
    \For{fold $k = 1, \ldots, K$}
        \State Train on $\mathcal{D}_{\text{train}}^{(k)} = \mathcal{D} \setminus \mathcal{D}_k$
        \State Obtain coefficients $\mathbf{c}^{(k)}(\theta)$ by solving Eq.~\ref{eq:pointwise_objective}
        \State Compute test log-likelihood:
        \State \quad $\mathcal{L}_k(\theta) = \sum_{i \in \mathcal{D}_k} \log f(z_i|\text{loc}_i; \theta, \mathbf{c}^{(k)})$
        \State Update: $\mathcal{L}_{\text{CV}}(\theta) \gets \mathcal{L}_{\text{CV}}(\theta) + \mathcal{L}_k(\theta)$
    \EndFor
\EndFor
\State \Return $\theta^* = \arg\max_{\theta \in \Theta} \mathcal{L}_{\text{CV}}(\theta)$
\end{algorithmic}
\end{algorithm}

\section{Posterior Variance, Fisher Information, and Prior Variance}
\label{supp:posterior_variance_section}

\paragraph{Symbol convention.}
To avoid collision with the mixture density $h_i = \alpha_i f_0(p_i) + (1-\alpha_i)f_1(p_i)$
of \S\ref{supp:natural_gradient}, we denote the per-observation Fisher information
by $\mathcal{I}_i$ throughout this section.

\subsection{Fisher Information of the Mixture}
\label{supp:fisher_connection}

\begin{proposition}[Fisher Information]
\label{prop:fisher_info}
For the mixture $f(p\mid\alpha) = \alpha f_0(p) + (1-\alpha) f_1(p)$, the observed Fisher information is:
\begin{equation}
    \mathcal{I}(\alpha;p)
    \;=\; \left(\frac{\partial \log f(p\mid\alpha)}{\partial\alpha}\right)^{\!2}
    \;=\; \frac{[f_0(p)-f_1(p)]^2}{[\alpha f_0(p)+(1-\alpha) f_1(p)]^2}.
\end{equation}
At observation $i$, we write $\mathcal{I}_i := \mathcal{I}(\alpha_i;p_i)$.
\end{proposition}

\begin{proof}
Direct computation:
\begin{equation}
\frac{\partial \log f}{\partial \alpha}
 \;=\; \frac{f_0(p) - f_1(p)}{\alpha f_0(p) + (1-\alpha) f_1(p)}
 \;\;\Longrightarrow\;\;
\mathcal{I}(\alpha;p)
 \;=\; \left(\frac{f_0(p) - f_1(p)}{\alpha f_0(p) + (1-\alpha) f_1(p)}\right)^{\!2}.
\end{equation}
\end{proof}

\begin{remark}[Fisher Information Behavior]
The Fisher information $\mathcal{I}_i$ exhibits characteristic behavior as a function of $\alpha$:
\begin{itemize}
    \item When $\alpha \approx 0.5$ and $f_0(p) \approx f_1(p)$: $\mathcal{I}_i \approx 0$ (low information, high statistical uncertainty).
    \item When $\alpha \approx 0$ or $\alpha \approx 1$: One distribution dominates, yielding larger $\mathcal{I}_i$ (more information, lower statistical uncertainty).
\end{itemize}
This creates the "Fisher floor" phenomenon: even with perfect spatial coverage, locations with $\alpha \approx 0.5$ have inherently high posterior uncertainty due to the mixture model ambiguity.
\end{remark}

\subsection{Hessian of the Centered Objective}

The Hessian of $\mathcal{L}(\mathbf{c})$ at $\hat{\mathbf{c}}$ is the second derivative of the three terms.

\paragraph{Data term.}
From $\mathcal{L}_{\text{data}} = -\sum_i \log h_i$ with $h_i = \alpha_i f_0(p_i)+(1-\alpha_i) f_1(p_i)$ and $\alpha_i = (K\mathbf{c})_i$, the second derivative gives
\begin{equation}
\nabla^2 \mathcal{L}_{\text{data}}
\;=\; K^T \mathrm{diag}(\mathcal{I}_i)\, K,
\end{equation}
where $\mathcal{I}_i$ is the per-observation Fisher information of Proposition~\ref{prop:fisher_info}.

\paragraph{Centered regularization term.}
For $\mathcal{L}_{\text{reg}} = \lambda_{\text{reg}}(\mathbf{c}-\mathbf{c}_{\bar\alpha})^T K (\mathbf{c}-\mathbf{c}_{\bar\alpha})$,
\begin{equation}
\nabla^2 \mathcal{L}_{\text{reg}} \;=\; 2\lambda_{\text{reg}} K.
\end{equation}
(Centering is a linear shift in $\mathbf{c}$ and does not affect curvature.)

\paragraph{Boundary term.}
With $\mathbf{b}_i = 2\cdot I(\alpha_i\notin[0,1])$,
\begin{equation}
\nabla^2 \mathcal{L}_{\text{bound}} \;=\; \lambda_{\text{bound}}\, K^T \mathrm{diag}(\mathbf{b})\, K.
\end{equation}

\paragraph{Total Hessian.}
\begin{equation}
\label{eq:hessian_centered}
H \;=\; K^T \mathrm{diag}(\mathcal{I}_i)\, K
   \;+\; 2\lambda_{\text{reg}}\, K
   \;+\; \lambda_{\text{bound}}\, K^T \mathrm{diag}(\mathbf{b})\, K.
\end{equation}

\subsection{Posterior Variance Theorem}
\label{supp:posterior_variance_theorem}

\begin{theorem}[Posterior Predictive Variance]
\label{thm:posterior_variance_full}
Under the Laplace approximation $\mathbf{c} \mid \text{data} \sim \mathcal{N}(\hat{\mathbf{c}}, H^{-1})$ with $H$ given by~\eqref{eq:hessian_centered}, the posterior variance of $\alpha(\text{loc})$ is:
\begin{equation}
\mathrm{Var}[\alpha(\text{loc}) \mid \text{data}] = \mathbf{k}(\text{loc})^T H^{-1} \mathbf{k}(\text{loc}),
\end{equation}
where $\mathbf{k}(\text{loc}) = [K(\text{loc}, \text{loc}_{s_1}), \ldots, K(\text{loc}, \text{loc}_{s_{N_0}})]^T$.
\end{theorem}

\begin{proof}
Since $\alpha(\text{loc}) = \mathbf{k}(\text{loc})^T \mathbf{c}$ is linear in $\mathbf{c} \sim \mathcal{N}(\hat{\mathbf{c}}, H^{-1})$:
\begin{equation}
\mathrm{Var}[\alpha(\text{loc})] = \mathbf{k}(\text{loc})^T \mathrm{Var}[\mathbf{c}] \mathbf{k}(\text{loc}) = \mathbf{k}(\text{loc})^T H^{-1} \mathbf{k}(\text{loc})
\end{equation}
by standard properties of linear transformations of Gaussian random vectors.
\end{proof}

\subsection{Properties of the Prior Prediction Variance}
\label{supp:prior_variance_properties}

\begin{proposition}[Properties of Prior Prediction Variance]
\label{prop:prior_variance_properties}
For a strictly positive definite kernel $K$ and regularization parameter $\lambda > 0$, the prior prediction variance:
\begin{equation}
\sigma^2_{\text{prior}}(\text{loc} \mid S) = K(\text{loc},\text{loc}) - \mathbf{k}(\text{loc})^T (K_S + \lambda I)^{-1} \mathbf{k}(\text{loc})
\end{equation}
satisfies:

\begin{enumerate}[label=(\roman*)]
    \item \textbf{Boundedness}: $0 \leq \sigma^2_{\text{prior}}(\text{loc} \mid S) \leq K(\text{loc},\text{loc})$ for all $\text{loc} \in \mathcal{H}$.

    \item \textbf{Monotonicity}: If $S_1 \subseteq S_2$, then $\sigma^2_{\text{prior}}(\text{loc} \mid S_2) \leq \sigma^2_{\text{prior}}(\text{loc} \mid S_1)$ for all $\text{loc} \in \mathcal{H}$.

    \item \textbf{Distance decay}: For stationary kernels $K(\text{loc}, \text{loc}') = k(\|\text{loc} - \text{loc}'\|)$ with $\lim_{r \to \infty} k(r) = 0$:
    \begin{equation}
    \lim_{\min_{s \in S} \|\text{loc} - \text{loc}_s\| \to \infty} \sigma^2_{\text{prior}}(\text{loc} \mid S) = K(\text{loc},\text{loc}).
    \end{equation}
\end{enumerate}
\end{proposition}

\begin{proof}
\textbf{(i) Boundedness.}
The lower bound follows from non-negativity of variance. For the upper bound, since $K_S + \lambda I \succ 0$, its inverse is positive definite. Therefore:
\begin{equation}
\mathbf{k}^T (K_S + \lambda I)^{-1} \mathbf{k} \geq 0 \implies \sigma^2_{\text{prior}} = K(\text{loc},\text{loc}) - \mathbf{k}^T (K_S + \lambda I)^{-1} \mathbf{k} \leq K(\text{loc},\text{loc}).
\end{equation}

\textbf{(ii) Monotonicity.}
Let $S_1 \subseteq S_2$ with $|S_1| = n$, $|S_2| = n + m$. Partition:
\begin{equation}
K_{S_2} = \begin{bmatrix} K_{S_1} & K_{12} \\ K_{21} & K_{22} \end{bmatrix},
\quad
\mathbf{k}_{S_2} = \begin{bmatrix} \mathbf{k}_{S_1} \\ \mathbf{k}_{\text{new}} \end{bmatrix}.
\end{equation}

By the Schur complement formula:
\begin{equation}
(K_{S_2} + \lambda I)^{-1} = \begin{bmatrix}
M & N \\ N^T & P
\end{bmatrix}
\end{equation}
where $M = (K_{S_1} + \lambda I)^{-1} + (K_{S_1} + \lambda I)^{-1}K_{12}P K_{21}(K_{S_1} + \lambda I)^{-1}$ and $P = [(K_{22} + \lambda I) - K_{21}(K_{S_1} + \lambda I)^{-1}K_{12}]^{-1} \succ 0$.

Computing the quadratic form,
\begin{align}
\mathbf{k}_{S_2}^T (K_{S_2}+\lambda I)^{-1} \mathbf{k}_{S_2}
&= \mathbf{k}_{S_1}^T (K_{S_1}+\lambda I)^{-1} \mathbf{k}_{S_1}
\;+\; v^T P\, v,
\end{align}
where $v = \mathbf{k}_{\text{new}} - K_{21}(K_{S_1}+\lambda I)^{-1} \mathbf{k}_{S_1}$. Since $P \succ 0$, $v^T P v \ge 0$, so
\begin{equation}
\mathbf{k}_{S_2}^T (K_{S_2}+\lambda I)^{-1} \mathbf{k}_{S_2}
\;\ge\; \mathbf{k}_{S_1}^T (K_{S_1}+\lambda I)^{-1} \mathbf{k}_{S_1},
\end{equation}
which by definition of $\sigma^2_{\text{prior}}$ yields $\sigma^2_{\text{prior}}(\text{loc} \mid S_2) \leq \sigma^2_{\text{prior}}(\text{loc} \mid S_1)$.

\textbf{(iii) Distance decay.}
For stationary kernels $K(\text{loc},\text{loc}') = k(\|\text{loc}-\text{loc}'\|)$ with $k(r)\to 0$, all entries of $\mathbf{k}(\text{loc})$ tend to $0$ as $\min_{s\in S}\|\text{loc}-\text{loc}_s\|\to\infty$. Since $K_S+\lambda I \succeq \lambda I$, the operator norm $\|(K_S+\lambda I)^{-1}\|_{\mathrm{op}} \le 1/\lambda$, so
\begin{equation}
\bigl|\mathbf{k}^T (K_S+\lambda I)^{-1}\mathbf{k}\bigr|
\;\le\; \|(K_S+\lambda I)^{-1}\|_{\mathrm{op}} \cdot \|\mathbf{k}\|^2
\;\le\; \frac{\|\mathbf{k}\|^2}{\lambda}
\;\to\; 0,
\end{equation}
hence $\sigma^2_{\text{prior}}(\text{loc}\mid S)\to K(\text{loc},\text{loc})$.
\end{proof}

\section{Density-Aware A-Optimal Design}
\label{supp:greedy_algorithm}
 
This section provides the algorithm referenced in Section~\ref{sec:entire_domain}
for selecting $N_0$ measurement locations that minimize average prediction
variance across the spatial domain.
 
\begin{algorithm}[H]
\caption{Density-Aware Greedy A-Optimal Location Selection}
\label{alg:greedy_design}
\begin{algorithmic}[1]
\Require Locations $\{\text{loc}_i\}_{i=1}^M$, kernel $K$, budget $N_0$, regularization $\lambda$, density parameter $\gamma$
\Ensure Selected subset $S$
\State Compute kernel matrix: $K_{\text{full}} \in \mathbb{R}^{M \times M}$
\State Compute local density: $\rho_i = \sum_{j \neq i} K(\text{loc}_i, \text{loc}_j)$ for all $i$
\State Compute weights: $w_i = \rho_i^\gamma / \sum_j \rho_j^\gamma$ \Comment{Normalize}
\State Initialize: $S \gets \emptyset$, $U \gets \{1,\ldots,M\}$
\For{$k = 1$ to $N_0$}
    \State $\text{best\_score} \gets +\infty$, $\text{best\_idx} \gets \text{null}$
    \For{each $j \in U$}
        \State $S_{\text{temp}} \gets S \cup \{j\}$
        \State $K_S \gets K_{\text{full}}[S_{\text{temp}}, S_{\text{temp}}]$
        \State $L \gets \text{cholesky}(K_S + \lambda I)$
        \State $\text{weighted\_var} \gets 0$
        \For{each $i \in U \setminus \{j\}$}
            \State $\mathbf{k}_i \gets K_{\text{full}}[i, S_{\text{temp}}]$
            \State $\mathbf{x}_i \gets L^{-T}(L^{-1}\mathbf{k}_i)$ \Comment{Solve using Cholesky}
            \State $\sigma^2_i \gets K_{\text{full}}[i,i] - \mathbf{k}_i^T \mathbf{x}_i$
            \State $\text{weighted\_var} \gets \text{weighted\_var} + w_i \cdot \sigma^2_i$
        \EndFor
        \If{$\text{weighted\_var} < \text{best\_score}$}
            \State $\text{best\_score} \gets \text{weighted\_var}$, $\text{best\_idx} \gets j$
        \EndIf
    \EndFor
    \State $S \gets S \cup \{\text{best\_idx}\}$, $U \gets U \setminus \{\text{best\_idx}\}$
\EndFor
\State \Return $S$
\end{algorithmic}
\end{algorithm}
 
\paragraph{Computational Complexity.} The algorithm performs $N_0$ iterations.
In iteration $k$, evaluating $(M-k)$ candidates each requiring Cholesky
factorization $O(k^3)$ and $(M-k)$ variance computations $O(k^2)$ gives
$O(M^2 k^2)$ per iteration. Total complexity: $O(M^2 N_0^3)$.
 
\paragraph{Density Parameter Selection.} The parameter $\gamma$ controls the
cluster-isolation tradeoff:
\begin{itemize}
    \item $\gamma > 0$: Prioritizes dense regions (population centers). Useful when inference targets are concentrated in clusters.
    \item $\gamma = 0$: Uniform weighting (balanced coverage). Recommended default for general applications.
    \item $\gamma < 0$: Prioritizes isolated points (rare events). Useful for outlier detection or when signals are expected in sparse regions.
\end{itemize}
 
For most applications, $\gamma \in [-0.5, 0.5]$ provides reasonable balance.
Cross-validation over $\gamma$ is possible but computationally expensive.

\section{Scalability Benchmarks}
\label{sec:scalability}
 
\begin{table}[h]
\centering
\caption{GPU benchmarks on a single RTX 3090 (24\,GB), \textbf{100 optimization iterations.}}
\label{tab:scalability}
\smallskip
\begin{tabular}{r l r r}
\toprule
$N$ & Method & Optimization (100 iter) & GPU Memory \\
\midrule
1{,}000   & Dense f32 & 0.04\,s   & 29\,MB  \\
5{,}000   & Dense f32 & 0.02\,s   & 513\,MB \\
10{,}000  & Dense f32 & 0.05\,s   & 2.0\,GB \\
20{,}000  & Dense f32 & 0.20\,s   & 8.0\,GB \\
30{,}000  & Dense f32 & 0.42\,s   & 18.0\,GB \\
40{,}000  & Dense f16 & 0.42\,s   & 16.0\,GB \\
50{,}000  & Dense f16 & 0.65\,s   & 15.0\,GB \\
75{,}000  & Tiled     & 33\,s     & 5.4\,GB \\
100{,}000 & Tiled     & 59\,s     & 5.4\,GB \\
200{,}000 & Tiled     & 3\,m\,58\,s  & 5.4\,GB \\
500{,}000 & Tiled     & 24\,m\,47\,s & 5.5\,GB \\
\bottomrule
\end{tabular}
\end{table}
 
\noindent\textbf{Dense float32} (full kernel in memory) handles $N \leq 30$K
with the complete pipeline, optimization, Hessian, and Cholesky, in under
4 seconds.
\textbf{Dense float16} halves memory, extending the dense method to $N = 50$K
with optimization still under 1 second.
\textbf{Tiled matrix-vector products} compute kernel tiles on-the-fly without
storing the full Gram matrix, keeping GPU memory flat at $\sim 5.4$\,GB
regardless of $N$, reaching $N = 500$K with \textbf{zero approximation error}.
The tradeoff is purely computational: tiled is $\sim 50\times$ slower than
dense at the same $N$, but any problem size fits on any GPU.

\end{document}